\documentclass[11pt]{article}
\usepackage[latin9]{inputenc}
\usepackage{geometry}
\usepackage{mathrsfs}
\usepackage{pdfpages}
\usepackage{amsmath}
\usepackage{amsthm}
\usepackage{setspace}
\usepackage[authoryear]{natbib}
\makeatletter

\theoremstyle{plain}
\newtheorem{assumption}{\protect\assumptionname}
\theoremstyle{definition}
\newtheorem*{defn*}{\protect\definitionname}
\theoremstyle{plain}
\newtheorem{lem}{\protect\lemmaname}
\theoremstyle{remark}
\newtheorem{rem}{\protect\remarkname}

\usepackage{amssymb}
\usepackage{titlesec}

\makeatother

\providecommand{\assumptionname}{Assumption}
\providecommand{\definitionname}{Definition}
\providecommand{\lemmaname}{Lemma}
\providecommand{\remarkname}{Remark}

\begin{document}
\title{Audits as Evidence: Experiments, Ensembles, and Enforcement\thanks{We thank Isaiah Andrews, Tim Armstrong, Kerwin Charles, Sendhil Mullainathan,
and Andres Santos for helpful conversations related to this project,
and John Nunley for providing data. We also thank participants at
the UC Berkeley labor and econometrics seminars, the Y-Rise External
Validity conference,  University of British Columbia, UC Irvine, Harris
School of Public Policy, Tinbergen Institute, Stockholm School of
Economics, and University of Oslo for useful feedback. Evan Rose and
Benjamin Scuderi provided outstanding research assistance. This project
was supported by a Russell Sage Foundation Presidential grant.}}
\author{Patrick Kline and Christopher Walters\\
 UC Berkeley and NBER}
\maketitle
\begin{abstract}
\begin{singlespace}
We develop tools for utilizing correspondence experiments to detect
illegal discrimination by individual employers. Employers violate
US employment law if their propensity to contact applicants depends
on protected characteristics such as race or sex. We establish identification
of higher moments of the causal effects of protected characteristics
on callback rates as a function of the number of fictitious applications
sent to each job ad. These moments are used to bound the fraction
of jobs that illegally discriminate. Applying our results to three
experimental datasets, we find evidence of significant employer heterogeneity
in discriminatory behavior, with the standard deviation of gaps in
job-specific callback probabilities across protected groups averaging
roughly twice the mean gap. In a recent experiment manipulating racially
distinctive names, we estimate that at least 85\% of jobs that contact
both of two white applications and neither of two black applications
are engaged in illegal discrimination. To assess more carefully the
tradeoff between type I and II errors presented by these behavioral
patterns, we consider the performance of a series of decision rules
for investigating suspicious callback behavior under a simple two-type
model that rationalizes the experimental data. Though, in our preferred
specification, only 17\% of employers are estimated to discriminate
on the basis of race, we find that an experiment sending 10 applications
to each job would enable accurate detection of 7-10\% of discriminators
while falsely accusing fewer than 0.2\% of non-discriminators. A minimax
decision rule acknowledging partial identification of the joint distribution
of callback rates yields higher error rates but more investigations
than our baseline two-type model. Our results suggest illegal labor
market discrimination can be reliably monitored with relatively small
modifications to existing audit designs.
\end{singlespace}
\end{abstract}
\bigskip{}


\bigskip{}
\bigskip{}
\bigskip{}

\pagebreak{}

\section{Introduction}

It is illegal to use information on age, race, or sex to make employment
decisions in the United States. Though economists are often called
upon to evaluate claims of illegal employment discrimination, academic
research in labor economics provides surprisingly little methodological
guidance for assessing whether particular employers are discriminating.
Rather, the focus of the voluminous empirical literature on labor
market discrimination \citep{altonji_blank,charles_guryan_jel} has
centered around methods for establishing whether \textit{markets}
discriminate against particular groups of workers on average. At least
since the work of \citet{becker_discrimination}, however, it has
been recognized that employers may vary substantially in the extent
to which they are discriminatory, and that this variation (in particular,
the difference in prejudice between the marginal and average firm)
influences the adverse impact of discrimination on outcomes for minority
workers \citep{charles_guryan_jpe}. It is therefore essential to
understand heterogeneity in discrimination across employers, both
for assessing the economic consequences of discrimination and for
enforcing anti-discrimination law.

This paper develops tools for detecting discrimination by individual
employers. The proposed methods rely on correspondence experiments
in which fictitious applications with randomly assigned characteristics
are submitted to actual job vacancies (\citealp{bertrand_duflo_review}
provide a review). A key advantage of correspondence studies over
traditional in-person audits is that the perceived traits of an applicant
can be independently manipulated, revealing the \emph{ceteris paribus}
influence of protected attributes such as race or gender on employer
behavior. Starting with the seminal work of \citet{bertrand_mullainathan_2004},
it has become standard to sample thousands of jobs and send each of
them four applications. Bertrand and Mullainathan found callback rates
to distinctively white names to be roughly 50\% higher than those
to distinctively black names, leading them to conclude that discrimination
was operative in the markets they studied. Our basic insight is that
such a study is best viewed as an \textit{ensemble} of many small
exchangeable experiments. From this ensemble, one can infer properties
of the distribution of discriminatory behavior which can, in turn,
be used to form posteriors about the probability that any given employer
is discriminating. These posteriors can then aid in making decisions
about which employers to investigate.

To ground our analysis, we develop a formal econometric framework
for analyzing correspondence studies. The foundation of this framework
is the assumption that callback outcomes at a particular job constitute
independent Bernoulli trials governed by job- and race-specific callback
probabilities. We show that this assumption is testable and document
empirical support for it in correspondence study data.  Treating the
pair of callback rates at each job as a random draw from a stable
super-population, we denote the joint cumulative distribution function
of white and black callback probabilities by $G\left(p_{w},p_{b}\right):\left[0,1\right]^{2}\rightarrow\left[0,1\right]$.
We then establish which moments of $G\left(\cdot,\cdot\right)$ are
identified as a function of the number of applications of each race
sent to each job. Though our focus is on correspondence studies, these
results are more broadly applicable to ensembles of randomized experiments
implemented across many sites.

Building on our identification results, we propose shape-constrained
Generalized Method of Moment (GMM) estimators of the identified moments
of the callback distribution that require the moment estimates be
rationalizable by a coherent bivariate probability distribution. We
apply these estimation methods to three experimental datasets: the
original \citet{bertrand_mullainathan_2004} study of racial discrimination,
a larger, more recent study by \citet{nunley_etal_2015} of racial
discrimination in the market for recent college graduates, and a study
by \citet{AGCV} that used eight applications per job. In each study,
we find overwhelming evidence of heterogeneity across jobs in the
extent of discrimination. In the more recent studies, where third
and higher moments are identified, we find evidence of skew and thick
tails in the distribution of discriminatory behavior: while most firms
barely discriminate, a few discriminate very heavily.

Next we consider what race-specific callback distributions $G\left(\cdot,\cdot\right)$
are consistent with the identified moments of the callback distribution.
Of particular interest is the fraction of jobs exhibiting any discrimination.
We derive an analytic lower bound on the fraction of jobs that engage
in discrimination conditional on the total number of callbacks. We
then show how sharp bounds can be computed via a linear programming
routine that works with a discrete approximation to $G\left(\cdot,\cdot\right)$
to characterize the relevant moment constraints. These bounds extend
some results in the literature on large scale inference concerned
with identification of the fraction of null hypotheses that are true
\citep{benjamini_hochberg,efron2001empirical,storey2002direct,efron2004large,efron2012large}.
We find that the linear programming bounds are significantly tighter
than our analytic bound and are informative even among the sub-population
of jobs calling back no applications. In the Bertrand and Mullanaithan
experiment, we estimate that at least half of the jobs calling back
one, two, or three of the four applications sent to each job are discriminating
based upon race. By contrast, as few as 20\% of the jobs that call
back all four applications are discriminating and as few as 5\% of
the jobs that call back no applications are discriminating.

These bounds on the fraction of jobs that discriminate constitute
a form of ``indirect evidence'' \citep{efron_2010} that can be
used to refine an assessment of whether any individual employer is
discriminating. Consider, for example, the case of a job sent four
applications that calls back only the two white applications. Under
the null hypothesis that callbacks do not depend on race at this job,
the probability of only the two white applications being contacted
given that two applications have been called back in total is $1/6$.
But in the Bertrand and Mullainathan data, we estimate that at most
56\% of the jobs that call back two applications in total are not
discriminating, so only a very weak presumption of innocence is justified.
Moreover, calling back only the white applications is relatively common,
occuring in 34\% of the cases where two total applications are called
back. Bayes' rule then implies the probability that such a job is
not discriminating is at most $\frac{1}{6}\times\frac{.56}{.34}\approx.27$.
Here, the indirect evidence tips the scale slightly in the employer's
favor but allows us to conclude that, at most, 27\% of such jobs are
not discriminating on the basis of race. This need not be the case
in general;  in the \citet{nunley_etal_2015} experiment, for example,
we estimate that at most 15\% of jobs that contact two white and zero
black applicants are not discriminating.

Making decisions based upon lower-bound posterior probabilities of
discrimination may yield overly conservative inferences. We develop
a decision theoretic framework formalizing the problem of a hypothetical
``auditor'' such as the Equal Employment Opportunity Commission
(EEOC) that has been charged with making decisions about whether to
investigate particular employers. Specifically, we consider Bayes
auditing rules under a linear loss function where each type I and
type II error incurs a fixed cost. The Bayes decision rule takes the
usual form, where investigations are conducted when the posterior
probability of discrimination crosses a fixed threshold. To approximate
the Bayes decision problem, we work with a mixed logit representation
of the data generating process (DGP) in the \citet{nunley_etal_2015}
experiment that also incorporates application characteristics other
than race. This model provides a good empirical fit to the \citet{nunley_etal_2015}
data  and  reproduces the qualitative patterns uncovered in our GMM
analysis. We use the mixed logit estimates to form empirical Bayes
posteriors of the probability that any given employer is discriminating.
For instance, the posterior probability that an employer that contacts
only the two white applications is discriminating is 62\%. But when
the two white applications have other characteristics indicating they
are of low quality and the two black applications have high-quality
characteristics, this posterior rises to 80\%.

We then compute the tradeoff between type I and type II errors implied
by the logit DGP under different experimental designs. With only two
white and two black applications per job, it is difficult to reliably
identify discriminating employers. But with just 10 applications per
job, we find that it is possible to correctly identify 7\% of discriminating
jobs with type I error rates of less than 0.2\%. Moreover, we show
that by optimizing over combinations of race and other resume characteristics
to maximize the amount of information generated by the experiment,
it is possible to boost the fraction of discriminators detected to
roughly 10\% while continuing to hold the type I error rate under
0.2\%. By contrast, conducting investigations based upon a frequentist
\emph{p-}value\emph{ }cutoff of 0.01 tends to yield substantially
more accusations, the majority of which are erroneous accusations
of non-discriminators.

Finally, we consider how our assessment of various auditing rules
changes if we relax the assumption that callbacks are actually generated
by the logit DGP and instead consider the broader class of callback
distributions capable of rationalizing the \citet{nunley_etal_2015}
experiment. A natural benchmark for decisionmaking with an unknown
risk function is the minimax auditing rule, which minimizes the maximum
risk that can arise given the identified moments of $G\left(\cdot,\cdot\right)$.
We develop a linear programming algorithm for computing an estimate
of the maximum risk function for classes of decision rules ordered
by their logit posteriors. Applying our algorithm, we find that the
gap between logit and worst case risk is decreasing in the share of
jobs investigated. This pattern leads a minimax auditor to investigate
more jobs than would an auditor who knew for certain that the logit
DGP governed behavior. The minimax auditor, it turns out, is more
concerned with the possibility that she is passing over a vast number
of jobs engaged in modest amounts of discrimination than that a few
non-discriminators are improperly investigated.

Our results highlight the potential of experimental methods to guide
regulatory enforcement. Because employers vary tremendously in their
propensity to discriminate against protected groups, regulators charged
with enforcing anti-discrimination laws face a difficult inferential
task. The methods developed here treat this task as an exercise in
large scale testing, which serves to discipline the conclusions drawn
regarding particular employers. Our findings suggest that accurately
monitoring illegal discrimination in online labor markets is feasible
with relatively small modifications to conventional audit designs.
Analogous methods are likely to be useful in other contexts in which
policymakers and researchers seek to draw conclusions about specific
individuals using noisy observations across many units. Candidates
for such applications include evaluations of teachers, schools, hospitals,
and neighborhoods \citep{chetty/friedman/rockoff:14b,ahpw_vam,hull_jmp,chettyhendren_neighborhoods_2,chetty_opportunity_atlas}.

The rest of the paper is structured as follows. The next section offers
a formal definition of employer discrimination in resume correspondence
experiments. Section 3 lays out the problem of an auditor seeking
to identify discriminatory employers with correspondence study data.
Section 4 develops identification results for moments of $G(\cdot,\cdot)$
and bounds on posterior probabilities of discrimination. Section 5
describes the data, and Section 6 uses it to test our modeling framework.
Section 7 provides estimated moments of $G(\cdot,\cdot)$ and Section
8 reports posterior bounds. Section 9 develops and estimates a mixed
logit model of callback decisions, and Section 10 uses the logit estimates
to evaluate prospects for detecting discrimination in alternative
experimental designs. Section 11 contrasts the logit results with
a minimax analysis acknowledging underidentification of $G(\cdot,\cdot)$.
Section 12 offers concluding thoughts.

\section{Defining Discrimination}

Title VII of the Civil Rights Act of 1964 prohibits employment discrimination
on the basis of race and sex, while the Age Discrimination Act of
1975 prohibits certain forms of discrimination on the basis of age.
Violations of these statutes typically involve one of two types of
claims. The first, \textit{disparate treatment}, consists of employment
practices that explicitly treat potential employees unequally based
upon protected characteristics. While economists have debated whether
such behavior arises from statistical discrimination versus racial
animus \citep{charles_guryan_jel}, the law makes no distinction between
these motives: both are clearly illegal \citep{kleinberg2019discrimination}.
The second sort of claim, \textit{disparate impact}, involves employment
practices that, while not explicitly based on protected characteristics,
clearly work to disadvantage members of such groups without offering
a corresponding productive justification.

Guided by these legal doctrines, we now develop a formal notion of
discrimination tailored to the analysis of correspondence studies.
To simplify exposition we focus on race, which we code as binary (``white''/
``black''), with the understanding that other protected characteristics
such as gender or age can play the same role. Suppose that we have
a sample of $J$ jobs with active vacancies. To each of these jobs,
we send $L_{w}$ applications with distinctively white names and $L_{b}$
applications with distinctively black names as in \citet{bertrand_mullainathan_2004},
for a total of $L=L_{w}+L_{b}$ applications. Denote the race associated
with the name used in application $\ell\in\left\{ 1,...,L\right\} $
to job $j\in\left\{ 1,...,J\right\} $ as $R_{j\ell}\in\left\{ w,b\right\} $.
Let the function $Y_{j\ell}\left(r\right):\left\{ w,b\right\} \rightarrow\left\{ 0,1\right\} $
denote whether job $j$ would call back application $\ell$ as a function
of that application's assigned race. Note that this definition of
potential outcomes builds in the Stable Unit Treatment Value Assumption
(SUTVA) of \citet{rubin_sutva} by ruling out dependence on the races
assigned to other resumes, including other applications to the same
job $\{R_{jk}\}_{k\neq\ell}$. Observed callbacks decisions are then
given by $Y_{j\ell}=Y_{j\ell}\left(R_{j\ell}\right)$.

When $Y_{j\ell}\left(w\right)\neq Y_{j\ell}\left(b\right)$ job $j$
has engaged in racial discrimination with application $\ell$. Notably,
even if racially distinctive names influence employer behavior only
through their role as a proxy for parental background \citep{fryer_levitt_2004},
using the names at any point in the hiring process is likely to be
viewed by courts as a pretext for discrimination (see, e.g., the discussion
in U.S. Equal Employment Opportunity Commission v. Target Corporation,
460 F.3d 946, 7th Cir. Wis. 2006\nocite{EEOCvTarget}). While courts
are typically interested in establishing whether a particular plaintiff
experienced discrimination in precisely this sense, we will take the
perspective of a regulator interested in assessing prospectively whether
an employer systematically treats applicants differently based upon
race. Such ``systematic'' forms of discrimination are also relevant
in establishing class certification in class action suits. Formalizing
this notion requires some additional assumptions regarding how employers
behave on average.

To this end, we work with the following representation of potential
outcomes:
\[
Y_{j\ell}(r)=Y_{j}(r,U_{j\ell}),
\]

\noindent where $Y_{j}(\cdot)$ is a job-specific decision rule and
$U_{j\ell}$ represents all factors that affect employer $j$'s decision
to contact application $\ell$ other than race. These factors include
the attributes other than race assigned to the resume as well as unobserved
fluctuations in conditions at the job such as changes in the time
available for reading applications. We normalize the marginal distribution
of $U_{j\ell}$ to be uniform, an assumption that is without loss
of generality since $Y_{j}(\cdot)$ is an unrestricted non-separable
function. Our first substantive assumption is that these unobserved
factors are independent across applications.
\begin{assumption}
\label{ass:1}The factors that influence callbacks are independent
and identically distributed across applications at each job:
\[
U_{j\ell}|R_{j1}...R_{jL}\overset{iid}{\sim}Uniform(0,1).
\]
\end{assumption}
\noindent Note that random assignment of racially distinctive names
to applications guarantees independence of $U_{j\ell}$ from $\{R_{jk}\}_{k=1}^{L}$.
The key behavioral restriction in Assumption 1 is that the $U_{j\ell}$
are mutually independent, which implies that each job's decision-making
is characterized by a stable decision rule applied independently to
each resume. This rules out, for example, a scenario in which the
job calls back the first qualified applicant and disregards all subsequent
applications. As we discuss later, this restriction turns out to be
testable, and we find that it provides a good empirical approximation
to behavior in the correspondence experiments we study.

Assumption \ref{ass:1} implies that each job is associated with a
stable pair of race-specific callback probabilities, defined as follows:
\[
p_{jr}\equiv\int_{0}^{1}Y_{j}\left(r,u\right)du,\ r\in\{b,w\}.
\]
The probability $p_{jr}$ may be interpreted as the callback rate
that would emerge in a hypothetical experiment in which a large number
of applications of race $r$ are sent to job $j$. Assumption 1 implies
that callbacks take the form of (race-specific) binomial trials governed
by these probabilities. Letting $C_{jr}=\sum_{\ell=1}^{L}1\left\{ R_{j\ell}=r\right\} Y_{j\ell}$
 denote the number of applications of race $r$ to job $j$ that were
called back, we can write the probability $\Pr\left(C_{jw}=c_{w},C_{jb}=c_{b}|p_{jw},p_{jb}\right)$
that employer $j$ calls back $c_{w}$ out of $L_{w}$ white applications
and $c_{b}$ out of $L_{b}$ black applications as:
\begin{eqnarray}
f\left(c_{w},c_{b}|p_{jw},p_{jb}\right) & = & \left(\begin{array}{c}
L_{w}\\
c_{w}
\end{array}\right)\left(\begin{array}{c}
L_{b}\\
c_{b}
\end{array}\right)p_{jw}^{c_{w}}\left(1-p_{jw}\right)^{L_{w}-c_{w}}p_{jb}^{c_{b}}\left(1-p_{jb}\right)^{L_{b}-c_{b}}.\label{eq:bin}
\end{eqnarray}

We are now ready to offer a definition of systematic discrimination,
which we will henceforth refer to simply as discrimination.
\begin{defn*}
\emph{Job $j$ engages in discrimination when $p_{jb}\neq p_{jw}$.}
\end{defn*}
\noindent We can now label discriminatory jobs with the indicator
function $D_{j}=1\{p_{jb}\neq p_{jw}\}$. Note that this definition
is prospective in that an employer with $D_{j}=1$ will eventually
discriminate against an applicant, though it may not do so in any
particular finite sample. Indeed, it is likely that many of the jobs
sampled in audit experiments are engaging in illegal discrimination
but have not discriminated against any of the fictitious applicants
in the study because none of the fictitious applicants would have
been called back regardless of their race.

\section{Ensembles and Decision Rules\label{sec:Ensembles}}

The above framework treats each job's callback decisions as a set
of race-specific Bernoulli trials. We next consider what can be learned
from a collection of experiments conducted at many jobs. This idea
is formalized in the following exchangeability assumption on the jobs.
\begin{assumption}
\label{ass:2}Race-specific callback probabilities are independent
and identically distributed:
\[
p_{jw},p_{jb}\overset{iid}{\sim}G\left(\cdot,\cdot\right).
\]
\end{assumption}
The distribution function $G\left(p_{w},p_{b}\right):\left[0,1\right]^{2}\rightarrow\left[0,1\right]$
describes the population of jobs from which a study samples. In practice,
audit studies usually draw small random samples of jobs from online
job boards. The \emph{iid} assumption abstracts from the fact that
there are a finite number of jobs on these boards. Note that by virtue
of random assignment (Assumption \ref{ass:1}) $p_{jw}$ and $p_{jb}$
are independent of the racial mix of applications to job $j$ as well
as any other resume characteristics that are randomized.

Assumption 2 implies that the unconditional distribution of callbacks
can be expressed as a mixture of binomial trials. Let $C_{j}\equiv\left(C_{jw},C_{jb}\right)$
denote the callback counts for job $j$. We denote the unconditional
probability of observing the callback vector $c=\left(c_{w},c_{b}\right)$
by 
\begin{eqnarray}
\bar{f}\left(c\right) & \equiv & \Pr\left(C_{j}=c\right)\nonumber \\
 & = & \int f\left(c_{w},c_{b}|p_{w},p_{b}\right)dG\left(p_{w},p_{b}\right).\label{eq: uncond}
\end{eqnarray}

The distribution $G\left(\cdot,\cdot\right)$ will serve as a key
object of interest in our analysis. One reason for interest in $G\left(\cdot,\cdot\right)$
is that it characterizes both the prevalence and extent of discrimination
in a population. For instance, the proportion of jobs that are \textit{not}
engaged in discrimination can be written:
\[
\pi^{0}\equiv\Pr\left(D_{j}=0\right)=\int dG\left(p,p\right).
\]

A second reason for interest in $G\left(\cdot,\cdot\right)$ lies
in its potential forensic value as a tool for identifying which jobs
are discriminating. The quantity 
\[
\pi\left(c\right)\equiv\Pr\left(D_{j}=1|C_{j}=c\right)
\]
gives the proportion of jobs with callback vector $c$ that are discriminating.
Though this quantity has a clear frequentist interpretation as the
fraction of discriminators that would be found under repeated sampling,
we can also think of it as giving a posterior probability of discrimination
given the ``evidence'' $C_{j}$. Specifically, invoking Bayes' rule,
we can write this posterior as a functional of the ``prior'' $G\left(\cdot,\cdot\right)$:
\begin{eqnarray*}
\pi\left(c\right) & = & \frac{\Pr\left(C_{j}=c|D_{j}=1\right)\left(1-\pi^{0}\right)}{\bar{f}\left(c\right)}\\
 & = & \frac{1-\pi^{0}}{\bar{f}\left(c\right)}\int_{p_{w}\neq p_{b}}f\left(c_{w},c_{b}|p_{w},p_{b}\right)dG\left(p_{w},p_{b}\right)\\
 & \equiv & \mathcal{P}\left(\underbrace{c}_{\text{direct}},\underbrace{G\left(\cdot,\cdot\right)}_{\text{indirect}}\right).
\end{eqnarray*}
The dependence of $\pi\left(c\right)$ on $G\left(\cdot,\cdot\right)$
is an example of what \citet{efron_2010} refers to as ``indirect
evidence.'' To understand the logic of incorporating indirect evidence,
suppose $\pi^{0}=1$ so that no jobs discriminate. Then $\pi\left(c\right)=0$
regardless of job $j$'s callback outcomes -- any seemingly suspicious
callback decisions are due to chance. Likewise, if $\pi^{0}=0$, all
jobs are discriminators and there is no need for direct evidence on
the behavior of particular jobs. But in intermediate cases, where
some fraction of jobs are discriminators, and some are not, it is
optimal to blend the direct evidence $C_{j}$ from a particular job
with contextual information on the population $G\left(\cdot,\cdot\right)$
from which that job was drawn to make decisions. We next  analyze
how exactly such indirect evidence should feature in decision-making
under uncertainty.

\subsection*{The Auditor's Problem}

Consider the problem of an auditor who knows the distribution $G\left(\cdot,\cdot\right)$
and aims to decide which jobs to investigate  for the presence of
illegal discrimination using a dataset of callbacks $\left\{ C_{j}\right\} _{j=1}^{J}$
as evidence. The auditor uses a deterministic decision rule $\delta\left(c\right):\left\{ 0,...,L_{w}\right\} \times\left\{ 0,...,L_{b}\right\} \rightarrow\left\{ 0,1\right\} $
that maps the callback vector $c=\left(c_{w},c_{b}\right)$ to a binary
inquiry decision.\footnote{We confine ourselves to deterministic rules because randomized decision
rules violate commonly held horizontal equity principles.}

The auditor's loss function from applying a decision rule $\delta\left(\cdot\right)$ to
a dataset of $J$ jobs is:
\begin{equation}
\mathcal{\mathcal{L}}_{J}\left(\delta\right)=\sum_{j=1}^{J}\left\{ \underbrace{\delta\left(C_{j}\right)\left(1-D_{j}\right)}_{\text{Type I}}\kappa+\underbrace{\left[1-\delta\left(C_{j}\right)\right]D_{j}}_{\text{Type II}}\gamma\right\} .\label{eq:loss}
\end{equation}
This loss function places a cost $\kappa$ on investigating ``innocent''
jobs with $D_{j}=0$ (type I errors) and a cost $\gamma$ of not investigating
``guilty'' jobs with $D_{j}=1$ (type II errors). Because the $\left\{ D_{j}\right\} _{j=1}^{J}$
are not known, the auditor minimizes expected loss (i.e. risk), which
we denote by $\mathcal{R}$:
\begin{eqnarray*}
\mathcal{R}_{J}\left(G,\delta\right) & \equiv & \mathbb{E}\left[\mathcal{L}_{J}\left(\delta\right)\right]\\
 & = & \sum_{j=1}^{J}\mathbb{E}\left[\delta\left(C_{j}\right)\left(1-\mathcal{P}\left(C_{j},G\right)\right)\kappa+\left[1-\delta\left(C_{j}\right)\right]\mathcal{P}\left(C_{j},G\right)\gamma\right]\\
 & = & J\mathbb{E}\left[\delta\left(C_{j}\right)\left(1-\mathcal{P}\left(C_{j},G\right)\right)\kappa+\left[1-\delta\left(C_{j}\right)\right]\mathcal{P}\left(C_{j},G\right)\gamma\right],
\end{eqnarray*}
where $\mathbb{E}\left[\cdot\right]$ denotes the expectation operator,
 the second line follows from iterated expectations, and the third
uses that the $\left\{ C_{j},D_{j}\right\} _{j=1}^{J}$ are $iid$
across jobs. The following Lemma, which mirrors a standard result
in statistical decision theory \citep[e.g.,][Theorem 8.11.1]{degroot2005optimal},
establishes that the optimal strategy of the auditor is to investigate
jobs that exceed a posterior threshold.
\begin{lem}[Posterior Threshold Rule]
The decision rule $\delta\left(C_{j}\right)=1\left\{ \mathcal{P}\left(C_{j},G\right)>\frac{\kappa}{\gamma+\kappa}\right\} $
minimizes $\mathcal{R}_{J}\left(G,\delta\right)$.\label{lem:lem1}
\end{lem}
\begin{proof}
Risk can be rewritten:
\begin{eqnarray*}
\mathcal{R}_{J}\left(G,\delta\right) & = & J\sum_{c_{w}=0}^{L_{w}}\sum_{c_{b}=0}^{L_{b}}\int\left\{ \delta\left(c_{w},c_{b}\right)\left(1-\mathcal{P}\left(c_{w},c_{b},G\right)\right)\kappa+\left[1-\delta\left(c_{w},c_{b}\right)\right]\mathcal{P}\left(c_{w},c_{b},G\right)\gamma\right\} \\
 & \times & f\left(c_{w},c_{b}|p_{w},p_{b}\right)dG\left(p_{w},p_{b}\right).
\end{eqnarray*}
Minimizing this integral pointwise, we see that for any $c=\left(c_{w},c_{b}\right)$
such that $\mathcal{P}\left(c,G\right)<\frac{\kappa}{\gamma+\kappa}$,
the integrand is minimized by setting $\delta\left(c\right)=0$. Otherwise,
risk is minimized by setting $\delta\left(c\right)=1$.
\end{proof}
One can think of the decision rule $\delta\left(C_{j}\right)$ as
offering an economically motivated definition of ``reasonable doubt'':
when the posterior probability of discriminating crosses the cost-based
threshold $\kappa/\left(\kappa+\gamma\right)$, it is rational to
conduct an inquiry.
\begin{rem}
Recent work by economists emphasizes the role of preferences in the
optimal design of experiments \citep{manski2000identification,kitagawa2018should,narita2019experiment}.
In our setting, an auditor might benefit from choosing the application
design $(L_{w,}L_{b})$ in addition to the decision rule $\delta(\cdot)$
to minimize risk. We consider such an exercise empirically in Section
\ref{sec:DET}.
\end{rem}

\subsection*{Connection to Large Scale Testing}

An interesting connection exists between the auditor's problem and
the literature on large scale testing, which is concerned with deciding
which hypotheses to reject based upon the results of a very large
number of tests (\citealp{efron2012large} provides a review). A seminal
contribution to this literature comes from \citet{benjamini_hochberg},
who proposed controlling the False Discovery Rate (FDR): the expected
share of rejected null hypotheses that are true. We next show that
the auditor's optimal decision rule will control an analogue of the
FDR.

Letting $N_{J}\equiv\sum_{j=1}^{J}\delta(C_{j})$ denote the total
number of investigations resulting from the auditing rule $\delta(\cdot)$,
we can define the \textit{Positive} \textit{False Discovery Rate}
\citep{storey2003positive} as:
\begin{center}
$pFDR_{J}=\mathbb{E}\left[N_{J}^{-1}\sum_{j=1}^{J}\delta(C_{j})(1-D_{j})|N_{J}\geq1\right]$.
\par\end{center}

\noindent \begin{flushleft}
In words, $pFDR_{J}$ gives the proportion of investigated jobs that
are not discriminating, conditional on at least one investigation
taking place. The following Lemma establishes that the optimal decision
rule controls $pFDR_{J}$ at a level determined by the ratio $\kappa/\gamma$.
\par\end{flushleft}
\begin{lem}[$pFDR_{J}$ Control]
\label{lem:FDR}If $\delta\left(C_{j}\right)=1\left\{ \mathcal{P}\left(C_{j},G\right)>\frac{\kappa}{\gamma+\kappa}\right\} $
then $pFDR_{J}\leq\frac{\gamma}{\kappa+\gamma}$.
\end{lem}
\begin{proof}
Storey \citeyearpar[Theorem 1]{storey2003positive} showed that $pFDR_{J}=\Pr(D_{j}=0|\delta(C_{j})=1)$
for any deterministic decision rule $\delta(\cdot)$ obeying $\Pr\left(\delta(C_{j})=1\right)>0$
(see Appendix \ref{sec:pFDR_appx} for a self-contained proof of this
result). Therefore the optimal auditing rule $\delta\left(C_{j}\right)=1\left\{ \mathcal{P}\left(C_{j},G\right)>\frac{\kappa}{\gamma+\kappa}\right\} $
yields
\begin{eqnarray*}
pFDR_{J} & = & \Pr\left(D_{j}=0|\mathcal{P}\left(C_{j},G\right)>\frac{\kappa}{\gamma+\kappa}\right)\\
 & \leq & \Pr\left(D_{j}=0|\mathcal{P}\left(C_{j},G\right)=\frac{\kappa}{\gamma+\kappa}\right)=1-\frac{\kappa}{\gamma+\kappa}.
\end{eqnarray*}
\end{proof}
By contrast, consider an auditor who bases investigations on a classical
hypothesis test $\delta^{\dagger}\left(C_{j}\right)$ that controls
size at a fixed level $\tilde{\alpha}<1$. To simplify exposition,
suppose that the test is pivotal under the null of non-discrimination
so that 
\[
\Pr\left(\delta^{\dagger}\left(C_{j}\right)=1|p_{jw}=p,p_{jb}=p\right)=\tilde{\alpha},\quad\forall p\in\left[0,1\right].
\]
We can write the resulting $pFDR_{J}$ of this rule
\begin{eqnarray*}
\Pr\left(D_{j}=0|\delta^{\dagger}\left(C_{j}\right)=1\right) & = & \frac{\Pr\left(\delta^{\dagger}\left(C_{j}\right)=1|D_{j}=0\right)\pi^{0}}{\Pr\left(\delta^{\dagger}\left(C_{j}\right)=1|D_{j}=0\right)\pi^{0}+\Pr\left(\delta^{\dagger}\left(C_{j}\right)=1|D_{j}=1\right)\left(1-\pi^{0}\right)}\\
 & \geq & \frac{\tilde{\alpha}\pi^{0}}{\tilde{\alpha}\pi^{0}+1-\pi^{0}}.
\end{eqnarray*}
To see that $\delta^{\dagger}\left(C_{j}\right)$ fails to control
$pFDR_{J}$, note that $\lim_{\pi^{0}\uparrow1}\frac{\tilde{\alpha}\pi^{0}}{\tilde{\alpha}\pi^{0}+1-\pi^{0}}=1$:
when nearly all jobs are innocent, classical hypothesis testing will
result in the vast majority of investigations being false accusations.
\begin{rem}
\label{rem:FDR}The \textit{False Discovery Rate} of \citet{benjamini_hochberg}
can be written $FDR_{J}=pFDR_{J}\times\Pr\left(N_{J}\geq1\right)$.
Because $\Pr\left(N_{J}\geq1\right)\leq1$, the optimal auditing rule
also controls $FDR_{J}$.
\end{rem}
\begin{rem}
The auditor's risk can be written 
\[
\mathcal{R}_{J}\left(G,\delta\right)=J\left\{ \kappa\times pFDR_{J}\times\Pr\left(\delta\left(C_{j}\right)=1\right)+\gamma\times pFNR_{J}\times\left[1-\Pr\left(\delta\left(C_{j}\right)=1\right)\right]\right\} ,
\]
where $pFNR_{J}=\mathbb{E}[\left(J-N_{J}\right)^{-1}\sum_{j=1}^{J}\left(1-\delta(C_{j})\right)D_{j}|N_{J}<J]$
is the \textit{Positive False Nondiscovery Rate} \citep[Corollary 4]{storey2003positive}.
Hence, the auditor's marginal rate of substitution between the Positive
False Discovery and Positive False Nondiscovery rates is $\frac{\kappa}{\gamma}\frac{\Pr\left(\delta\left(C_{1}\right)=1\right)}{1-\Pr\left(\delta\left(C_{1}\right)=1\right)}$.
\end{rem}

\subsection*{Auditing under Ambiguity}

The distribution $G\left(\cdot,\cdot\right)$ will not, in general,
be point identified even by experiments with many applications per
job. When $G$ is only known to lie in some identified set $\Theta$
of distributions, many possible decision rules are consistent with
rationality. Among those rules, an important benchmark is the minimax
decision rule \citep{wald1945statistical,savage1951theory,manski2000identification},
which minimizes the maximum risk that may arise from the experiment.
We define the maximum risk function and the associated minimax decision
rule respectively as:
\begin{equation}
\mathcal{R}_{J}^{m}\left(\Theta,\delta\right)\equiv\sup_{G\in\Theta}\mathcal{R}_{J}\left(G,\delta\right)\quad\text{and}\quad\delta^{mm}\equiv\arg\inf_{\delta\in\mathscr{D}}\mathcal{R}_{J}^{m}(\Theta,\delta),\label{eq:max_risk}
\end{equation}
where $\mathscr{D}$ is the set of deterministic decision rules.

Unlike in the case where $G\left(\cdot,\cdot\right)$ is known, an
auditor that only knows $G\in\Theta$ cannot consult a single posterior
probability to make the decision of whether to investigate. Rather,
the maximum risk of each decision rule must be computed to obtain
the minimax decision rule. The next section establishes more carefully
what features of $G\left(\cdot,\cdot\right)$ are identified by a
given experimental design and provides an approach to computing $\mathcal{R}^{m}\left(\Theta,\delta\right)$.
\begin{rem}
Rules that minimize maximum risk over a restricted set $\Gamma$ of
distributions were considered by \citet{hodges1952use} and \citet{robbins1964empirical}
and are sometimes referred to as $\Gamma-$minimax estimators \citep[see, e.g.,][]{berger1979gamma,noubiap2001algorithm,lehmann2006theory,berger2013statistical}.
While the statistics literature has typically chosen the set of candidate
distributions based upon prior beliefs, the definition in (\ref{eq:max_risk})
restricts consideration to distributions that match the identified
features of $G(\cdot,\cdot)$.
\end{rem}
\begin{rem}
The minimax decision rule $\delta^{mm}$ is also a Bayes rule if there
is a $\bar{G}\in\varTheta$ such that $\arg\inf_{\delta\in\mathscr{D}}\mathcal{R}_{J}\left(\bar{G},\delta\right)=\delta^{mm}$.
When such a $\bar{G}$ exists, we call it a least favorable distribution
because $\bar{G}\in\arg\sup_{G\in\Theta}\mathcal{R}_{J}\left(G,\delta^{mm}\right)$
\citep[see, e.g.,][]{lehmann2006theory}.
\end{rem}
\begin{rem}
One can think of the minimax decision rule $\delta^{mm}$ as an estimator
of the latent discrimination indicators $\left\{ D_{j}\right\} $$_{j=1}^{J}$.
It is interesting to contrast $\delta^{mm}$ with standard ``shrinkage''
estimators, which are typically motivated by appeal to a parametric
mixing distribution that serves the role of a prior \citep{kane/staiger:08,chetty/friedman/rockoff:14a,ahpw_vam,chettyhendren_neighborhoods_2,finkelstein_nejm_shrinkage}.
When it has a Bayes interpretation, the minimax estimator $\delta^{mm}$
can be thought of as shrinking towards a least favorable prior distribution
$\bar{G}$.
\end{rem}

\section{Identification of $G$}

In this section we establish that certain moments of $G\left(\cdot,\cdot\right)$
are non-parametrically identified and then proceed to derive bounds
on the posterior probability function $\pi\left(c\right)$.

\subsection*{Moments}

From (\ref{eq: uncond}) we can write
\begin{center}
$\bar{f}(c_{w},c_{b})=\left(\begin{array}{c}
L_{w}\\
c_{w}
\end{array}\right)\left(\begin{array}{c}
L_{b}\\
c_{b}
\end{array}\right)\mathbb{E}\left[p_{jw}^{c_{w}}\left(1-p_{jw}\right)^{L_{w}-c_{w}}p_{jb}^{c_{b}}\left(1-p_{jb}\right)^{L_{b}-c_{b}}\right]$
\begin{equation}
=\left(\begin{array}{c}
L_{w}\\
c_{w}
\end{array}\right)\left(\begin{array}{c}
L_{b}\\
c_{b}
\end{array}\right)\sum_{x=0}^{L_{w}-c_{w}}\sum_{s=0}^{L_{b}-c_{w}}(-1)^{x+s}\left(\begin{array}{c}
L_{w}-c_{w}\\
x
\end{array}\right)\left(\begin{array}{c}
L_{b}-c_{b}\\
s
\end{array}\right)\mathbb{E}\left[p_{jw}^{c_{w}+x}p_{jb}^{c_{b}+s}\right].\label{eq: mom}
\end{equation}
\par\end{center}

\begin{flushleft}
Hence, the reduced form callback probabilities can be written as linear
functions of uncentered moments $\mu\left(m,n\right)\equiv\mathbb{E}\left[p_{jw}^{m}p_{jb}^{n}\right]=\int p_{w}^{m}p_{b}^{n}dG\left(p_{w},p_{b}\right)$
of the latent callback probabilities.
\par\end{flushleft}

Letting $\bar{f}=\left(\bar{f}\left(1,0\right),...,\bar{f}\left(L_{w},0\right),...,\bar{f}\left(L_{w},L_{b}\right)\right)'$
denote the vector of frequencies for all possible callback outcomes
excluding $(0,0)$ and $\mu=\left(\mu\left(1,0\right),...,\mu\left(L_{w},0\right),...,\mu\left(L_{w},L_{b}\right)\right)'$
the corresponding list of moments, we can write the equations in (\ref{eq: mom})
as a linear system:
\[
\bar{f}=B\mu,
\]
where $B$ is a known non-singular square matrix of binomial coefficients.
We can therefore solve the linear system as:
\begin{equation}
\mu=B^{-1}\bar{f}.\label{eq:inverse}
\end{equation}
Hence, for a given application design $(L_{w},L_{b})$, all moments
$\mu(m,n)$ for $0\leq m\leq L_{w}$ and $0\leq n\leq L_{b}$ are
identified.

From $\mu$ we can compute centered moments of the callback distribution,
which are typically easier to interpret. An example of particular
interest is the standard deviation of discrimination across jobs:
\begin{eqnarray*}
\mathbb{V}\left[p_{jb}-p_{jw}\right]^{1/2} & = & \sqrt{\mathbb{E}\left[\left(p_{jb}-p_{jw}\right)^{2}\right]-\mathbb{E}\left[p_{jb}-p_{jw}\right]^{2}}\\
 & = & \sqrt{\mu\left(0,2\right)+\mu\left(2,0\right)-2\mu\left(1,1\right)-\mu\left(0,1\right)^{2}-\mu\left(1,0\right)^{2}+2\mu\left(0,1\right)\mu\left(1,0\right)}.
\end{eqnarray*}
This quantity is identified for any application design that sends
at least two resumes per racial group ($\min\left\{ L_{w},L_{b}\right\} \geq2$).
\begin{rem}
The variance of unit-specific treatment effects is typically treated
as under-identified in standard analyses of randomized experiments
\citep{neyman_1923,heckman_smith_clements,imbens_rubin_2015}. Identification
is secured here by the assumption that callbacks are generated by
race-specific binomial trials with stable probabilities, essentially
allowing repeated observations on the same units with and without
treatment.
\end{rem}
\begin{rem}
In experiments where the application design ($L_{w},L_{b})$ varies
randomly across jobs, some moments of $G(\cdot,\cdot)$ are over-identified.
We exploit these over-identifying restrictions in estimation to improve
precision and test our modeling assumptions.
\end{rem}

\subsection*{A Bound on Posterior Probabilities}

Though the study of moments of the callback distribution $G(\cdot,\cdot)$
can shed light on underlying heterogeneity in callback behavior, the
posterior probability $\pi\left(c\right)$ need not admit a representation
in terms of a finite number of moments. However, a simple analytic
bound on the posterior can be derived from an application of Bayes
rule that conditions on the total number of callbacks $C_{jb}+C_{jw}$
to firm $j$. Let $\bar{f}_{t}\left(c\right)\equiv\Pr\left((C_{jw},C_{jb})=(c_{w},c_{b})|C_{jb}+C_{jw}=t\right)$
denote the probability mass function for callbacks in the stratum
of jobs that call back $t$ applicants in total. Similarly, let $\pi_{t}^{0}\equiv\Pr\left(D_{j}=0|C_{jb}+C_{jw}=t\right)$
denote the share of jobs in this stratum that are  not discriminating.
Letting $t=c_{w}+c_{b}$, we can then write the posterior probability
of innocence as:
\begin{equation}
1-\pi\left(c\right)=\Pr\left(C_{j}=c|D_{j}=0,C_{jb}+C_{jw}=t\right)\frac{\pi_{t}^{0}}{\bar{f}_{t}\left(c\right)}.\label{eq: Bayes_cond}
\end{equation}

Note that after conditioning on the total number of callbacks, the
callback likelihood for non-discriminators is free of nuisance parameters:
\[
\Pr\left(C_{j}=c|D_{j}=0,C_{jb}+C_{jw}=t\right)=\left(\begin{array}{c}
L_{w}\\
c_{w}
\end{array}\right)\left(\begin{array}{c}
L_{b}\\
c_{b}
\end{array}\right)/\left(\begin{array}{c}
L\\
t
\end{array}\right)\equiv\bar{f}_{t}^{0}\left(c\right).
\]
This ratio forms the basis of Fisher's exact test for independence
in contingency tables \citep{fisher_1922}. For example, with two
white and two black applications and two callbacks, the probability
of both callbacks being to white applications under the null hypothesis
of non-discrimination is
\begin{center}
$\bar{f}_{2}^{0}\left(2,0\right)=\left(\begin{array}{c}
2\\
2
\end{array}\right)\left(\begin{array}{c}
2\\
0
\end{array}\right)/\left(\begin{array}{c}
4\\
2
\end{array}\right)=\dfrac{1}{6}.$
\par\end{center}

Since the function $\bar{f}_{t}\left(c\right)$ is directly identified
by random sampling, the sole under-identified quantity in equation
(\ref{eq: Bayes_cond}) is $\pi_{t}^{0}$, which serves as the auditor's
prior probability that an employer is innocent knowing only that it
made $t$ callbacks in total. The following Lemma provides a tractable
bound on this quantity.
\begin{lem}[Upper Bound on Stratum Prior]
\label{lem:analytical_bound} $\pi_{t}^{0}\leq\min_{c:c_{w}+c_{b}=t}\min\left\{ \frac{\bar{f}_{t}\left(c\right)}{\bar{f}_{t}^{0}\left(c\right)},\frac{1-\bar{f}_{t}\left(c\right)}{1-\bar{f}_{t}^{0}\left(c\right)}\right\} .$
\end{lem}
\begin{proof}
By the law of total probability:
\[
\bar{f}_{t}\left(c\right)=\bar{f}_{t}^{0}\left(c\right)\pi_{t}^{0}+\bar{f}_{t}^{1}\left(c\right)\left(1-\pi_{t}^{0}\right),
\]
where $\bar{f}_{t}^{1}\left(c\right)\equiv\Pr\left(C_{j}=c|D_{j}=1,C_{jb}+C_{jw}=t\right)$.
The result follows immediately from observing that $\bar{f}_{t}^{1}\left(c\right)\in\left[0,1\right]$.
\end{proof}
Plugging the upper bound of Lemma \ref{lem:analytical_bound} into
(\ref{eq: Bayes_cond}) therefore yields a \textit{lower} bound on
the posterior probability of discrimination: 
\begin{equation}
\pi\left(c\right)\geq1-\frac{\bar{f}_{t}^{0}\left(c\right)}{\bar{f}_{t}\left(c\right)}\min_{c^{\prime}:c_{w}^{\prime}+c_{b}^{\prime}=t}\min\left\{ \frac{\bar{f}_{t}\left(c^{\prime}\right)}{\bar{f}_{t}^{0}\left(c^{\prime}\right)},\frac{1-\bar{f}_{t}\left(c^{\prime}\right)}{1-\bar{f}_{t}^{0}\left(c^{\prime}\right)}\right\} .\label{eq: Bayes_lower}
\end{equation}

\begin{rem}
A bound of the sort derived in Lemma \ref{lem:analytical_bound} was
used by \citet[p. 1154]{efron2001empirical} to control $FDR_{J}$
in a multiple testing analysis of a microarray experiment. \citet{storey2002direct}
proposed a related class of upper bounds that are generally looser,
but easier to estimate (see \citealp{armstrong2015adaptive} for an
approach to inference on these bounds).
\end{rem}

\subsection*{Sharp Bounds}

While the bounds in Lemma \ref{lem:analytical_bound} are easy to
compute, they need not be sharp, as restrictions across strata defined
by the number of callbacks $C_{jb}+C_{jw}$ have been ignored. An
upper bound on the prior $\pi_{t}^{0}$ that exploits all of the logical
restrictions in our framework can be written as the solution to the
following constrained optimization problem:
\begin{equation}
\max_{G\left(\cdot,\cdot\right)\in\mathscr{G}}\frac{\left(\begin{array}{c}
L\\
t
\end{array}\right)}{\sum_{\left(c_{w}',c_{b}'\right):c_{w}'+c_{b}'=t}\bar{f}\left(c_{w}',c_{b}'\right)}\int p^{t}\left(1-p\right)^{L-t}dG\left(p,p\right),\label{eq:LP}
\end{equation}

\begin{align}
s.t.\enskip\bar{f}\left(c_{w},c_{b}\right) & =\left(\begin{array}{c}
L_{w}\\
c_{w}
\end{array}\right)\left(\begin{array}{c}
L_{b}\\
c_{b}
\end{array}\right)\int p_{w}^{c_{w}}\left(1-p_{w}\right)^{L_{w}-c_{w}}p_{b}^{c_{b}}\left(1-p_{b}\right)^{L_{b}-c_{b}}dG\left(p_{w},p_{b}\right),\label{eq:LPcons}\\
\text{for} & \left(c_{w}=0,..,L_{w};c_{b}=0,..,L_{b}\right).\nonumber 
\end{align}

To make this problem computationally tractable, we consider a space
$\mathscr{G}$ of discretized approximations to the unknown distribution
function $G\left(\cdot,\cdot\right)$ (see \citealp{noubiap2001algorithm}
and \citealp{tebaldi_torgovitsky_yang} for related approaches). Because
both the objective and constraints are linear in the probability mass
function associated with $G\left(\cdot,\cdot\right)$, we can apply
linear programming routines to compute bounds given an estimate of
the callback probabilities $\left\{ \bar{f}\left(c_{w},c_{b}\right)\right\} _{c_{w},c_{b}}$
(we defer the discussion of estimation to Section \ref{sec:Moment-Estimates}).
Details of our computational procedure are given in  Appendix \ref{sec:LP_appx}.
\begin{rem}
Unlike the conditional bounds of Lemma \ref{lem:analytical_bound},
the solution to (\ref{eq:LP}) enforces constraints across callback
strata. As a result, we can obtain informative bounds on the fraction
of discriminatory jobs even among those that call no applications
back.
\end{rem}
\begin{rem}
Analogous linear programming formulations can be used to bound any
linear functional of $G(\cdot,\cdot)$, including other measures of
discrimination. For example, we can bound from below the fraction
of employers discriminating against whites by replacing the objective
in (\ref{eq:LP}) with $\min_{G\left(\cdot,\cdot\right)\in\mathscr{G}}\int_{p_{w}<p_{b}}dG\left(p_{w},p_{b}\right)$.
We leverage this insight to bound a variety of features of $G(\cdot,\cdot)$
in the empirical work to follow.
\end{rem}

\subsection*{Maximum risk}

Relying on a discretized function space $\mathscr{G}$ also enables
us to compute the maximum risk function $\mathcal{R}_{J}^{m}$ consistent
with a set of experimental callback probabilities.\footnote{See \citet{muller2019nearly} for a closely related approach to minimizing
weighted average risk involving discretization of unbiasedness constraints.} For a given decision rule $\delta\left(\cdot\right)$, $\mathcal{R}_{J}^{m}\left(\delta\right)$
can be expressed as the solution to the following optimization problem:
\begin{eqnarray*}
\mathcal{R}_{J}^{m}(\delta) & = & \max_{G\in\mathscr{G}}J\sum_{c_{w}=1}^{L_{w}}\sum_{c_{b}=1}^{L_{b}}\int\left\{ \delta\left(c_{w},c_{b}\right)1\left\{ p_{w}=p_{b}\right\} \kappa+\left[1-\delta\left(c_{w},c_{b}\right)\right]1\left\{ p_{w}\neq p_{b}\right\} \gamma\right\} \\
 & \times & f\left(c_{w},c_{b}|p_{w},p_{b}\right)dG\left(p_{w},p_{b}\right)\quad s.t.\ \left(\ref{eq:LPcons}\right).
\end{eqnarray*}
When $\mathscr{G}$ is a family of discrete distributions, the objective
and constraints are both linear in the probability masses associated
with $G\left(\cdot,\cdot\right)$ and $\mathcal{R}_{J}^{m}\left(\delta\right)$
can be computed numerically as the solution to a linear programming
problem. The minimax decision rule $\delta^{mm}\left(\cdot\right)$
can be found by computing $\mathcal{R}_{J}^{m}(\delta)$ for each
candidate rule $\delta\in\mathscr{D}$ and choosing the rule that
yields lowest maximal risk.
\begin{rem}
With $L_{w}$ white and $L_{b}$ black applications there are $2^{\left(1+L_{w}\right)\left(1+L_{b}\right)}$
distinct rules to consider which, in practice, prohibits brute force
enumeration when $L_{w}+L_{b}>4$. To circumvent this obstacle, we
consider in our empirical application a restricted set $\mathscr{D}^{\dagger}\subset\mathscr{D}$
of decision rules that rely upon a nominal ordering of $\mathscr{D}$
and a threshold rule.
\end{rem}

\section{Data}

We apply our methods to data from three correspondence experiments
summarized in Table I. Bertrand and Mullainathan (BM, 2004) applied
to 1,112 job openings in Boston and Chicago, submitting four applications
to each job. Of the four applications, two were assigned black-sounding
names while the remaining two were assigned white-sounding names.
The callback rate to applications with black sounding names was 3.1
percentage points lower than to applications with white sounding names.

Nunley et al. (2015) studied racial discrimination in the market for
new college graduates by applying to 2,305 listings on an online job
board, again sending four resumes per job opening. Unlike Bertrand
and Mullainathan, the names assigned to the four resumes were sampled
without replacement from a pool of eight names, four of which were
distinctively black and four of which were distinctively white. This
led the fraction of names sent to each job that were distinctively
black to vary randomly in increments of 25\% from 0\% to 100\%. Interestingly,
the average callback rate in the Nunley study was more than twice
as high as in the Bertrand and Mullainathan study, perhaps because
the fictitious applicants were more highly educated. On average, black
sounding names had a 2.6 percentage point lower callback rate than
white names.

Arceo-Gomez and Campos-Vasquez (AGCV, 2014) applied to 802 job openings
through an online job portal in a study of race and gender discrimination
in Mexico City, Mexico. They sent eight fictitious applications to
each job, and the applicants were all recent college graduates. For
simplicity, we focus on gender in this experiment, as AGCV used a
three-category definition of race that is more complicated to analyze
and may be less generalizable to other settings. In their data, women
are 3.4 percentage points more likely to receive callbacks than men.
While the Arceo-Gomez and Campos-Vasquez (2014) experiment looks at
a very different context than Bertrand and Mullainathan (2004) and
Nunley et al. (2015), this data set allows us to demonstrate the gains
from doubling the number of applications per job opening.

\section{Are Callbacks Rival?}

We begin by considering tests of the binomial trials assumption (Assumption
1) that undergirds our econometric framework. Fundamentally, we are
concerned that the probability of application $\ell$ receiving a
callback from job $j$ might depend not only on its own characteristics
but the characteristics of the other applications sent to it. To assess
this possibility, we fit linear probability models of the form:
\begin{equation}
Y_{j\ell}=\lambda_{0}+X_{j\ell}'\lambda_{1}+\bar{X}_{j\ell}'\lambda_{2}+\varepsilon_{j\ell},\label{eq:peer}
\end{equation}
where $X_{j\ell}$ is a vector of application characteristics and
$\bar{X}_{j\ell}=\left(L-1\right)^{-1}\sum_{k\neq\ell}X_{jk}$ gives
the ``leave out'' mean of those characteristics among the applications
sent to job $j$ excluding application $\ell$. While the coefficient
vector $\lambda_{1}$ gives the direct effect of application characteristics
on callbacks, the coefficient vector $\lambda_{2}$ captures the ``peer
effect'' of other applications to the same job on application $\ell$'s
callback propensity. Assumption \ref{ass:1} restricts these peer
effects to be zero ($\lambda_{2}=0$).

For OLS estimates of (\ref{eq:peer}) to identify a causal effect
of $\bar{X}_{j\ell}$, we need for $\bar{X}_{j\ell}$ to be uncorrelated
with any omitted application characteristics $Z_{j\ell}$ that influence
callbacks. In Bertrand and Mullainathan (2004)'s study, this condition
is violated because of two features of the design. First, application
characteristics were assigned according to their joint distribution
in a training sample, making it likely that $X_{j\ell}$ and $Z_{j\ell}$
are correlated. Second, the application characteristics were chosen
to yield a good match with the job (see pp. 996), leading $Z_{j\ell}$
to be correlated with its leave out mean $\bar{Z}_{j\ell}$ and hence
with $\bar{X}_{j\ell}$. For this reason, we focus on the Nunley et
al. (2015) study which assigned both race and application characteristics
independently of each other and across applications. Note that even
when this independence holds, $X_{j\ell}$ and $\bar{X}_{j\ell}$
will tend to be negatively correlated when $L$ is small \citep{angrist_peers}.
Omitting $X_{j\ell}$ from (\ref{eq:peer}) would therefore tend to
lead to a spurious finding of negative peer effects. So long as $X_{j\ell}$
is included, however, a finding that $\lambda_{2}\neq0$ provides
evidence of a peer effect.

Table II reports estimates of the parameters in (\ref{eq:peer})
for the Nunley et al. (2015) study, with each row showing the coefficients
from a separate regression. While applications with distinctively
black names are significantly less likely to be called back, we find
no significant effect on callback probabilities of changing the racial
mix of the other 3 applications to the same job. In fact, the point
estimate indicates that increasing the share of applicants with distinctively
black names insignificantly lowers callback rates, which is the opposite
of what one would expect if callbacks were rival. Across the 12 covariates
we consider only one (an indicator for 3+ months of unemployment)
finds a significant peer effect at conventional levels. However, a
joint test fails to reject that all of the leave out mean coefficients
are jointly zero.

As another composite test, we report the results of a model in which
the peer effects are restricted to be proportional to the main effects
of the application's own characteristics $X_{j\ell}$. The row titled
``predicted callback rate'' pools all the application characteristics
into an index $X_{j\ell}\hat{\lambda}_{1\left(j\right)}$ where $\hat{\lambda}_{1\left(j\right)}$
is the leave out OLS coefficient vector obtained from regressing the
callback indicator on application covariates after leaving out all
applications to job $j$. In the Nunley et al. (2015) study, a unit
increase in $X_{j\ell}\hat{\lambda}_{1\left(j\right)}$ is associated
with roughly half of a callback on average. Note that if we had used
the leave-in OLS prediction $X_{j\ell}\hat{\lambda}_{1}$ as the regressor,
this coefficient would mechanically equal one. Though $X_{j\ell}\hat{\lambda}_{1\left(j\right)}$
strongly predicts callbacks, its average value among competing applications
$\left(L-1\right)^{-1}\sum_{k\neq\ell}X_{jk}\hat{\lambda}_{1\left(j\right)}$
has no statistically discernible impact on callbacks.

Columns 3 and 4 of Table II report corresponding estimates for the
AGCV data. Unfortunately, the covariates in this dataset have insignificant
direct effects, so tests of indirect effects are fundamentally under-powered.
Empirically then Assumption 1 seems to provide a good approximation
to callback behavior in the Nunley et al. (2015) experiment. A possible
explanation for this result may be that these researchers applied
to posted vacancies where employers were capable of hiring multiple
applicants to the same job. Whatever the reason, we now proceed to
estimating the moments of the callback distribution based on Assumption
1.

\section{Moment Estimates\label{sec:Moment-Estimates}}

Tables III-V report estimates of identified moments of the callback
distribution in each of our three experiments. We report method of
moments estimates that apply equation (\ref{eq:inverse}) to the sample
callback frequencies as well as ``shape constrained'' GMM estimates
that require the callback frequencies be rationalizable by a proper
(discretized) probability distribution $G\in\mathscr{G}$. Imposing
shape constraints serves two goals. First, we need the moment estimates
to be rationalizable by some $G\in\mathscr{G}$ in order to subsequently
use them as constraints when estimating bounds via our linear programming
method. Second, when the constraints bind, the resulting estimates
are typically closer to the truth and more precise (see \citealp{chetverikov2018econometrics}
for a review). Details of the shape constrained estimation procedure
appear in Appendix \ref{sec:shapeconstrained_appx}. Table VI uses
the shape constrained estimates to summarize key features of the distribution
of callback probabilities in each experiment.

\subsection*{\citet{bertrand_mullainathan_2004}}

Estimates of centered moments in the Bertrand and Mullainathan experiment
are shown in Table III. Because this study employed a single design
with $L_{w}=L_{b}=2$, the moments reported are just identified. The
first column of Table III reports method of moments estimates with
standard errors computed via the delta method. The first row of the
Table shows the mean callback probabilities of white and black applications
across jobs which, because of the balanced application design, match
the callback rates reported in Table I. More interesting are the second
moments: there is substantial over-dispersion in callback probabilities,
with standard deviations across jobs for each race-specific probability
more than double the mean probability. As expected, there is also
a strong positive covariance between white and black callback rates,
reflecting that some employers simply call back more applications
of all types.

As shown in column 2, the shape constraints do not bind in the Bertrand
and Mullainathan data, which means the sample frequencies can be rationalized
to numerical precision by a discretized probability distribution.
Consequently, the resulting moment estimates are identical to the
method of moments estimates of column 1. Though the constraints do
not bind, it is hypothetically possible for the variability of the
constrained estimator to be lower if some of the constraints are near-binding.
However, our standard error estimates, which rely on the ``numerical
bootstrap'' procedure of \citet{hong2017numerical} (described in
Appendix \ref{sec:shapeconstrained_appx}), suggest that the constrained
GMM estimator is roughly as variable as the unconstrained method of
moments estimator.

Columns 1-3 of Table VI report transformations of the moments in Table
III that are somewhat easier to interpret. Most notably, we find substantial
heterogeneity in the difference in race specific callback rates $p_{jb}-p_{jw}$
across jobs. The standard deviation of the job-specific causal effect
is more than twice as large as the mean effect. The third row of Table
VI shows the correlation between white and black callback probabilities
is very large, at 0.927. However, the correlation between the discriminatory
gap in callback rates $p_{jb}-p_{jw}$ and the white callback probability
$p_{jw}$ is strongly negative, suggesting that discrimination tends
to be stronger when firms have higher chances of calling back more
white workers. This reflects, in part, a mechanical boundary effect,
as an employer with very low callback rates has little opportunity
to discriminate. Since the white callback rate in this study is only
around 10\%, boundary effects are likely to be a quantitatively important
phenomenon.

\subsection*{\citet{nunley_etal_2015}}

Moment estimates from the \citet{nunley_etal_2015} study are reported
in Table IV. Recall that \citet{nunley_etal_2015} employed five distinct
application designs with $\left(L_{jw},L_{jb}\right)\in\{\left(4,0\right),\left(3,1\right),\left(2,2\right),\left(1,3\right),\left(0,4\right)\}$.
Columns 1-3 of Table IV report design-specific method of moments estimates
of all identified moments for the three designs with the largest sample
sizes.\footnote{The remaining designs were omitted from this analysis due to small
sample sizes. Only 22 jobs were in the $\left(L_{jw}=0,L_{jb}=4\right)$
design while 43 jobs fell in the $\left(L_{jw}=4,L_{jb}=0\right)$
design.} As expected, the design-specific estimates are generally close to
one another. The sole moment that appears to differ across designs
is the mean white callback rate, which is somewhat lower in the $\left(1,3\right)$
design than the $\left(3,1\right)$ design.  However, column 5 of
Table IV shows that the differences between the designs are not jointly
statistically significant, which is in line with our findings from
the previous section.

To pool the designs efficiently, we again use a shape constrained
GMM estimator that requires the moments be rationalizable by a proper
probability distribution $G\in\mathscr{G}$. The pooled estimates,
reported in column 5 of Table IV, lie closest to those from the $\left(2,2\right)$
design, which has the largest sample size. Moments identified solely
by the $\left(1,3\right)$ and $\left(3,1\right)$ designs change
more substantially when pooling across designs, as the binomial structure
of the probabilities imposes restrictions across moments. As usual,
the minimized value of our GMM criterion function provides a measure
of the goodness of fit of our model. Applying the bootstrap method
of \citet{chernozhukov2015constrained} yields a \emph{p}-value of
0.19 for the null hypothesis that the results for all experimental
designs are jointly rationalized by the model. As expected, pooling
the designs substantially improves the precision of the estimated
coefficients. Notably, the standard error estimates fall substantially
even for many just-identified moments due to the cross-moment restrictions
implied by the model.

Consistent with our findings for the Bertrand and Mullainathan data,
columns 4-6 of Table VI reveal substantial heterogeneity in race-specific
callback rates in the Nunley et al. (2015) experiment, with standard
deviations roughly twice their mean. The imbalanced design used by
Nunley et al. (2015) allows us to identify higher moments than the
earlier Bertrand and Mullainathan study despite the total number of
applications sent being the same. While race-specific callback rates
are right skewed, racial gaps in callback probabilities $p_{jb}-p_{jw}$
are left-skewed, suggesting a long tail of heavy discriminators.

\subsection*{\citet{AGCV}}

Column 1 of Table V reports just-identified method of moments estimates
for the AGCV data. Column 2 imposes shape constraints on the moments,
which bind strongly in this case, presumably because the design of
the AGCV experiment involves many small cells. Despite substantial
movement in the moment estimates, the bootstrap \emph{p}-value on
the null hypothesis that the callback frequencies are generated by
the model is 0.79, indicating that the raw callback frequencies are
rationalizable by a well-behaved underlying joint distribution of
callback probabilities. Moreover, imposing the shape constraints reduces
the estimated standard errors of some of these moments. It is important
to note, however, that the asymptotic distribution of the shape constrained
estimator will tend to be non-normal \citep{fang2018inference} and
so standard errors provide only a heuristic guide to the uncertainty
associated with each moment estimate.

Columns 7-9 of Table VI report key moment estimates from the AGCV
data. The behavior of the first two moments is similar to that reported
in the prior two experiments, with gender-specific standard deviations
roughly twice their mean callback probabilities. However, the greater
number of applications used in this design helps enormously with the
precision of higher moment estimates.\footnote{Though the standard errors reported in Table VI suggest imprecision
in our estimates of the higher moments of the female callback rate
distribution, this appears to be a consequence of the asymptotic non-normality
of the shape-constrained estimator. For example, the numerical bootstrap
gives a 90-percent confidence interval of {[}5.37, 7.49{]} for the
excess kurtosis of $p_{jf}$ while the corresponding standard error
equals 8.79.} We find strong evidence of left-skew in the distribution of gender
gaps in callback probabilities as well as evidence of excess kurtosis
in the distribution of gaps. While many jobs discriminate little,
there is a thick tail of heavy discriminators.

\section{Posterior Bounds}

Our analysis of moments revealed substantial heterogeneity in callback
probabilities and discrimination across employers. Next, we compute
lower bound estimates of the probability that a given employer is
discriminating. In computing both the analytic bounds of Lemma \ref{lem:analytical_bound}
and the sharp bounds of (\ref{eq:LP}), we replace the unknown callback
probabilities $\bar{f}$ with estimates $\hat{\bar{f}}=B\hat{\mu}$,
where $\hat{\mu}$ is the relevant vector of shape constrained moment
estimates reported in Tables III-V. To ensure our bounds are not artificially
tight, our linear programming algorithm employs a grid with 36 times
as many points as the grid used in our earlier GMM step.

\subsection*{\citet{bertrand_mullainathan_2004}}

Table VII reports upper bounds on the fraction of jobs that are not
engaged in discrimination by the number of applications called back
in the Bertrand and Mullainathan experiment. Column 1 of Table VII
reports estimates of the analytic bounds in Lemma \ref{lem:analytical_bound}:
at most 62\% of the jobs that call back 2 applications are innocent
of discrimination, while at most 56\% of jobs that call back 3 applications
are not discriminating. Column 2 of Table VII reports estimates of
the sharp linear programming bounds. The sharp upper bounds are somewhat
lower than their analytical counterparts, revealing that at most 56\%
of the jobs calling back two applicants are not discriminating. Among
jobs that call back three applications, at most half are not discriminating
on the basis of race. In this callback stratum, our estimates suggest
jobs should not logically be presumed innocent.

The linear programming approach also generates informative bounds
in callback strata for which analytical bounds are not available.
Overall, at most 87\% of jobs do not discriminate on the basis of
race. Notably, at most 96\% of jobs that call back no applications
are not engaged in discrimination, while at most 79\% of jobs that
call back all four applications do not discriminate on the basis of
race. Since neither of these strata exhibit any difference in black-white
callback rates, all of the relevant information on discrimination
in these strata comes from the total number of callbacks blended with
the indirect evidence from the callback distribution $G\left(\cdot,\cdot\right)$.

Column 3 of Table VII reports linear programming-based upper bounds
on the proportion of jobs with white callback probabilities greater
than or equal to their black callback probability, $\Pr(p_{wj}\geq p_{bj})$.
We find an upper bound of exactly one in each callback stratum, indicating
that the callback probabilities can be rationalized without any employers
engaging in ``reverse discrimination'' against whites. Column 4
of Table VII reports upper bounds on the proportion of jobs with white
callback probabilities less than or equal to their black callback
probabilities, $\Pr(p_{wj}\leq p_{bj})$. These upper bound estimates
coincide exactly with those reported in column 2. Accordingly, we
easily reject the null hypothesis of no discrimination against blacks.

Figure I converts the upper bound estimates in column 2 of Table VII
to lower bound posterior probabilities of discrimination. Overall,
at least 13\% of jobs engage in discrimination. However, at least
72\% of jobs that call back two white and no black applications are
discriminating, while a job that calls back one white and no black
applications has at least a 58\% chance of discriminating. Highlighting
the role of indirect evidence, we estimate that at least 4\% of jobs
that call back no applicants and at least 21\% of jobs that call back
all applicants discriminate on the basis of race.

\subsection*{\citet{nunley_etal_2015}}

Table VIII reports upper bound estimates of the probability of innocence
from the \citet{nunley_etal_2015} study for each application design
involving both races. In column 1, our analytic bound formula suggests
at most 72\% of the jobs calling back two applicants in a balanced
design with $L_{jw}=L_{jb}=2$ are not discriminating -- slightly
higher than the corresponding estimate in Bertrand and Mullainathan.
This upper bound is higher in the two imbalanced designs $\left(L_{jw}=3,L_{jb}=1\right)$
and $\left(L_{jw}=1,L_{jb}=3\right)$.

Applying the linear programming  approach tightens these bounds dramatically
and provides additional bounds on the prevalence of discrimination
among jobs that make no callbacks or that call every application.
We estimate that at most 64\% of all jobs have equal white and black
callback probabilities, with that share falling to under 31\% among
employers who call back two applicants in a balanced $\left(2,2\right)$
design. However, some of this discrimination is estimated to be against
whites. Column 3 shows that our shape constrained callback probabilities
$\hat{\bar{f}}$ imply that at most 85\% of employers have white callback
probabilities greater than or equal to black probabilities. However,
these moments are estimated with error, and a bootstrap test of the
null hypothesis that all employers have white callback probabilities
weakly exceeding their black callback probability yields a \emph{p}-value
of 0.12. If we attribute the evidence of reverse discrimination to
sampling error, we can take the estimates in column 3 as the relevant
upper bounds on non-discrimination, which are closer to the analytical
bounds reported in column 1. Column 4 of Table VIII reports that at
most 83\% of jobs have white callback probabilities less than or equal
to black probabilities. Unsurprisingly, we decisively reject the null
hypothesis that this upper bound is one, indicating that discrimination
against blacks is substantial.

Figure II converts the upper bound priors reported in column 3 of
Table VIII into posterior estimates of the share of employers with
selected callback configurations engaged in discrimination against
blacks. Overall, at least 15\% of jobs discriminate against blacks
(i.e., have $p_{jw}>p_{jb}$). However, we estimate that at least
85\% of the employers calling back two white and no black applicants
in a balanced $\left(2,2\right)$ design are discriminating against
blacks. Interestingly, calling back three whites and no blacks in
a $(3,1)$ design is estimated to be even more suspicious, with at
least 90\% of the employers generating this callback evidence engaged
in discrimination against blacks.

\subsection*{\citet{AGCV}}

Table IX reports upper bound estimates of the probability of innocence
in the AGCV experiment. Focusing on the sharp bounds reported in column
2, we find that at most 72\% of jobs are not engaged in discrimination
against either gender. Remarkably, this share falls to 11\% among
jobs calling back a single applicant and rises to only 28\% among
jobs calling back two applicants. This bound is much lower than the
analytic bound in column 1, showing that cross-stratum restrictions
in a design with eight applications are very useful for tightening
bounds in strata with few callbacks. Evidently, jobs that call back
few applicants in the AGCV experiment are very likely to engage in
discrimination.

Some of this discrimination appears to be ``reverse'' discrimination
against women. Column 3 shows that at most 91\% of jobs do not discriminate
against women and a bootstrap test of the null hypothesis that this
bound equals one is decisively rejected. An employer that calls back
a single application has at most a 59\% chance of not discriminating
against women. Column 4 shows that at most 81\% of jobs do not discriminate
against men, and our bootstrap \emph{p}-value indicates this bound
is also statistically distinguishable from one. The mean difference
in callback rates in the ACGV experiment therefore masks gender discrimination
operating in both directions. An employer that calls back a single
application has at most a 52\% chance of not discriminating against
men.

Figure III plots lower bound posterior probabilities of discrimination
against men and women, respectively, for selected callback configurations.
Unconditionally, at least 20\% of jobs discriminate against men (i.e.,
have $p_{jm}<p_{jf}$), while at least 10\% of jobs discriminate against
women (i.e., have $p_{jf}<p_{jm}$). At least 97\% of the jobs that
call back four women and no men are estimated to discriminate against
men. But even an employer that calls back a single woman and no men
has at least a 90\% chance of discriminating against men. Likewise,
at least 85\% of jobs that call back a single man and no women are
estimated to be discriminating against women. Note that the under
null of non-discrimination, the probability of a particular gender
being contacted given a single callback in total is $\bar{f}_{1}^{0}\left(1,0\right)=\bar{f}_{1}^{0}\left(0,1\right)=1/2$.
That we obtain such strikingly informative posteriors in settings
with a single callback demonstrates the tremendous value of indirect
evidence in this setting.

\section{Parametric Models}

The previous section demonstrated that standard audit experiments
allow robust non-parametric inferences to be drawn about the discriminatory
status of particular jobs. In this section, we contrast the non-parametric
bounding methods developed above with the results of considering a
simple  parametric family $\mathbb{G}_{\theta}$ of distributions
for $G\left(\cdot,\cdot\right)$. We estimate the parameter vector
$\theta$ by maximum likelihood. If the true $G\left(\cdot,\cdot\right)$
lies in $\mathbb{G_{\theta}}$ then this approach will yield consistent
and efficient estimates of $\theta$, while if the model is misspecified,
maximum likelihood will still provide an approximation to whatever
features of $G\left(\cdot,\cdot\right)$ are identified. Parametric
modeling also facilitates incorporating other application characteristics
into the callback probabilities. This can serve to generate more nuanced
posteriors; for example, an employer that calls back both of two low
quality white applications but neither of two high quality black applications
is particularly suspicious.

We work with a mixed logit model of the form
\[
\Pr\left(Y_{j\ell}=1|R_{j\ell},X_{j\ell},\alpha_{j},\beta_{j}\right)=\Lambda\left(\alpha_{j}-\beta_{j}1\left\{ R_{j\ell}=b\right\} +X_{j\ell}'\psi\right),
\]
where $\Lambda\left(\cdot\right)=\frac{\exp\left(\cdot\right)}{1+\exp\left(\cdot\right)}$
is the standard logistic CDF, $X_{j\ell}$ is a vector of de-meaned
application covariates, and $\left(\alpha_{j},\beta_{j}\right)$ are
random coefficients governing the odds of a white callback and discrimination
against blacks respectively. To allow for heterogeneity in white callback
rates we assume that $\alpha_{j}\overset{iid}{\sim}N\left(\alpha_{0},\sigma_{\alpha}^{2}\right)$.
Discrimination is modeled as a two-type (conditional) mixture:

\[
\beta_{j}|\alpha_{j}=\begin{cases}
\beta_{0} & \text{w/ prob. \ensuremath{\Lambda\left(\tau_{0}+\tau_{\alpha}\alpha_{j}\right)}},\\
0 & \text{w/ prob. \ensuremath{1-\Lambda\left(\tau_{0}+\tau_{\alpha}\alpha_{j}\right)}}.
\end{cases}
\]

The above specification allows for some fraction of jobs to not discriminate
at all, while the remaining jobs depress the odds of calling back
blacks relative to whites by roughly $\beta_{0}\%$. When $\tau_{\alpha}\neq0$,
the probability of a job discriminating depends on $\alpha_{j}$,
which governs the white callback rate. Note that random assignment
of the covariates $X_{j\ell}$ implies they can safely be excluded
from the type probability equation.

\subsection*{Model Estimates}

Table X shows the results of fitting the above model to the \citet{nunley_etal_2015}
experiment. Column 1 provides a standard ``random effects'' logit
model with heterogeneity confined to the intercept as in \citet{farber_etal_audit}.
We find substantial variability across jobs in the overall odds of
a callback: a 0.1 standard deviation increase in the intercept $\alpha_{j}$
is estimated to raise the odds of a callback by 47\%. We also find
clear evidence of market-wide discrimination: black applications have
roughly 46\% lower odds of being called back than their white counterparts.

Column 2 allows the race effect $\beta_{j}$ to vary across employers,
which yields a significant improvement in model fit. The types specification
finds that only about 17\% of jobs discriminate against blacks --
very near the lower bound estimate of 15\% produced by our linear
programming routine (see column 3 of Table VIII). However, the degree
of discrimination among such jobs is estimated to be very severe --
the odds of receiving a call back are roughly $\exp\left(4\right)-1\approx53$
times higher for white applications than blacks. Column 3 allows the
probability of being the discriminatory type to depend on the white
callback rate, which yields a negligible improvement in model fit.
Surprisingly, $\alpha_{j}$ and $\beta_{j}$ are found to be nearly
independent, which implies that the negative correlation between $p_{jb}-p_{jw}$
and $p_{jw}$ reported in Table VI is attributable to boundary effects.
Again, this model finds roughly 17\% of jobs discriminate against
blacks.

Because we cannot reject the null hypothesis that $\tau_{\alpha}=0$
(i.e. that discrimination is independent of white callback rates),
we work with the more parsimonious model in column 2 in the exercises
that follow. Figure IV provides a goodness of fit diagnostic for this
model, plotting the empirical callback rates in each black / white
callback by application design cell against the logit model's predicted
callback probability in that cell. The empirical frequencies track
the model predictions closely and a naive Pearson $\chi^{2}$ test
fails to reject the null hypothesis that the model rationalizes the
cell frequencies up to sampling error.

\subsection*{Posteriors}

Figure V reports the distribution of posterior probabilities $\Pr(D_{j}=1|\{Y_{j\ell},R_{j\ell},X_{j\ell}'\psi\}_{\ell=1}^{L})$
implied by the parameter estimates reported in column 2 of Table X.
To summarize the influence of the covariates, we evaluate the posteriors
at two points within each race group, corresponding to the estimated
index $X_{j\ell}'\hat{\psi}$ being a standard deviation above or
below its empirical mean, which we refer to as ``high'' and ``low''
quality applications respectively. By construction, the mean posterior
coincides with the fraction of jobs that are estimated to be discriminating.
The types model finds that only 17\% of jobs are discriminating, yielding
a strong prior of innocence. Calling back only white applicants still
justifies a substantial degree of suspicion, however: 62\% of the
jobs that call back two whites and no blacks are discriminating. Imbalances
in the covariate mix of applicants can substantially intensify this
suspicion. Specifically, 79\% of the jobs that call back two low quality
white applications and neither of two high quality black applications
are discriminating. Evidently, even in models with a strong presumption
of innocence, four applications can provide enough information to
cast substantial doubt on whether individual employers are in compliance
with US employment law. However, it is only under the most extreme
callback configurations that we can detect discriminators with reasonable
certainty. In the next section, we study more carefully the tradeoff
between type I and II errors presented by the two-type model, and
how that tradeoff evolves with the number of applications sent.

\section{Experimental Design and Detection Error Tradeoffs\label{sec:DET}}

We now study the ability of a hypothetical auditor to detect employer
discrimination in a hypothetical population characterized by the two-type
logit model of callback rates reported in column 2 of Table X. This
parameterization found that discrimination was very rare, with only
17\% of jobs engaging in discrimination. It is of considerable interest
then to understand how the tradeoff between type I (``false accusation'')
and type II (``false acquittal'') errors faced by an auditor changes
with the number of applications sent to each job.

Recall from Lemma \ref{lem:lem1} that an auditor's optimal decision
rule is to investigate when the posterior probability of discrimination
crosses a cost-based threshold. We presume the auditor forms posteriors
using her prior knowledge of $G\left(\cdot,\cdot\right)$ which coincides
with the two-type estimates reported in Table X. This corresponds
to a thought experiment in which the auditor fits the two-type model
to the Nunley et al. (2015) experiment, forms consistent estimates
of the logit parameters, and then sends applications to additional
vacancies posted on the same online job board from which the original
study sampled.

Figure VI displays a rescaling of the type I and II error rates that
arise from implementing decision rules corresponding to varying posterior
thresholds. The horizontal axis of Figure I gives the share of jobs
engaged in discriminating that are investigated (``accused''). The
vertical axis plots the share of non-discriminators that are not investigated
(i.e., that are ``acquitted''). Note that this quantity corresponds
to $1-(FDR_{1}/\pi^{0})$ (see Remark \ref{rem:FDR}). Each point
gives the values of these shares corresponding to a particular posterior
decision threshold. The bold point corresponds to a posterior threshold
of 80\%.

In the canonical design with only 4 applications (2 white and 2 black),
the 80\% posterior threshold yields almost no false accusations. This
control over type I errors comes at the cost of a very high type II
error rate -- very few accusations of any sort are made, leading
to a negligible fraction of discriminators being detected. Note that
conducting a classical hypothesis test (e.g., Fisher's exact test)
at the 1\% level is equivalent to controlling the fraction of correct
acquittals, which is depicted by the horizontal line at 0.99. This
rule would yield more accusations but most of those accusations would
be erroneous -- the equivalent posterior threshold in the 2 pair
design is only about 33\%.

Expanding the design to 5 pairs of applications (5 white and 5 black)
yields a very substantial outward shift in the detection error tradeoff
curve. Using a posterior threshold of 80\% keeps the fraction of employers
falsely accused of discrimination below 0.2\% while allowing detection
of roughly 7.5\% of the jobs that are actually discriminating. Evidently,
ten applications is enough to accurately detect a non-trivial fraction
of discriminators.

The third line probes the potential for further gains by choosing
the racial mix and covariates of the applications optimally. Specifically,
we consider the set of 10-application portfolios generated by all
possible combinations of race and two covariate-based application
``quality'' bins, and select the portfolio that minimizes risk for
a given posterior threshold. We then vary that threshold to trace
out the detection error tradeoff. Choosing applications optimally
yields modest improvements in type I and II error rates. Using an
80\% posterior threshold, the share of non-discriminators investigated
remains below 0.2\%, while the share of discriminators investigated
rises to roughly 10\%.

Interestingly, the risk minimizing portfolio for an auditor with an
80\% posterior threshold consists of 5 high quality black applications
and 5 low quality white applications. Intuitively, an employer that
calls back low quality white applications more often than high quality
black applications is very likely to be discriminating against blacks.
Of course, the results of such an experiment would be difficult to
interpret without prior knowledge of $G\left(\cdot,\cdot\right)$,
as one would not be able to parse the separate effects of race and
quality. Consequently, the gains from optimizing the portfolio of
applications are in practice only achievable in a sequential experiment
in which a first wave is used to estimate $G\left(\cdot,\cdot\right)$.

\section{Ambiguity and Auditing Thresholds}

The analysis of the previous section suggested that a favorable mix
of type I and type II errors can be achieved when 10 applications
are sent to each job and jobs with a posterior probability of discriminating
greater than 80\% are investigated. However, that analysis assumed
that $G\left(\cdot,\cdot\right)$ was characterized by the parameters
of our two-type logit model, which provides only one of many possible
rationalizations of the moments identified by the \citet{nunley_etal_2015}
experiment. To assess how our hypothetical auditor's risk might change
under different distributional assumptions, and how best to respond
to this ambiguity, we now study the maximum risk function $\mathcal{R}_{J}^{m}(\delta)$.
In this exercise, we assume each job $j$ is characterized by a tuple
$\left(p_{jw}^{H},p_{jw}^{L},p_{jb}^{H},p_{jb}^{L}\right)$ of race
by quality callback probabilities drawn from the joint distribution
$G\left(p_{w}^{H},p_{w}^{L},p_{b}^{H},p_{b}^{L}\right)$. For comparison
with the logit model, which only allowed discrimination against blacks,
we define discrimination as occuring when callback probabilities are
\textit{higher} for whites within either quality stratum; that is,
we define $D_{j}=1-1\{p_{jb}^{H}\geq p_{jw}^{H}\}1\{p_{jb}^{L}\geq p_{jw}^{L}\}$.\footnote{Defining discrimination more narrowly as any difference between white
and black callback probabilities within either quality stratum yields
nearly identical results.}

To facilitate comparison with the logit model, we consider a restricted
family $\mathcal{D}^{\dagger}$ of decision rules  of the form $\delta\left(C_{j},q\right)=1\left\{ \mathcal{P}\left(C_{j},X'\psi,G_{logit}\right)\geq q\right\} ,$
where $q\in(0,1)$ is a cutoff and $G_{logit}$ is the logit model
reported in column 2 of Table X.\footnote{In cases where multiple evidence configurations yield posterior $q$,
we consider separate rules that investigate each of these configurations
individually.} Computing the maximal risk for this family of decision rules can
be thought of as a way of ``second guessing'' the risk associated
with each logit posterior threshold without debating the logit model's
ordering of the underlying evidence configurations. In computing $\mathcal{R}_{J}^{m}(\delta)$,
we use the logit predictions of $\bar{f}\left(c_{w},c_{b}\right)$
within each of the two quality bins of $X'\hat{\psi}$ as constraints
(see Appendix \ref{sec:maxrisk_appx} for details) and choose loss
parameters $\kappa=4$ and $\gamma=1$ so that, under the logit DGP,
an 80\% posterior threshold would minimize risk.

Figure VII plots average risk $\mathcal{R}_{J}^{m}\left(\delta\left(\cdot,q\right)\right)/J$
against the nominal logit posterior threshold $q$. For each decision
rule, the maximal risk is much higher than the average logit risk,
with the ratio between the two risks growing (discontinuously in some
cases) with the posterior cutoff. As the posterior threshold approaches
one -- so that no jobs are accused -- the maximal risk approaches
one because the least favorable $G\left(\cdot,\cdot\right)$ entails
every job being guilty. Conversely as the posterior threshold approaches
zero -- so that all jobs are accused -- the maximum risk approaches
four because the least favorable $G\left(\cdot,\cdot\right)$ is one
where nearly all jobs are innocent. Recall however from Table VIII
that not all jobs can be innocent in the \citet{nunley_etal_2015}
experiment, which is why $\sup_{\delta\in\mathcal{D}^{\dagger}}\mathcal{R}_{J}^{m}\left(\delta\right)/J$
is a value less than four.

While the logit risk function $\mathcal{R}_{J}\left(G_{logit},\delta\right)$
is minimized by the decision rule with a threshold nearest 80\%, $\mathcal{R}_{J}^{m}\left(\delta\right)$
is minimized by a rule with an implicit (logit-based) threshold of
only 18\%. This lower threshold implies a minimax auditor would investigate
many more jobs than an auditor with the same preferences who knows
$G\left(\cdot,\cdot\right)$ to be logit. Evidently, the minimax auditor
is more concerned with the possibility that she is passing over a
vast number of jobs engaged in modest amounts of discrimination than
that a few non-discriminators are improperly investigated. To gain
some intuition for this result, note that the minimax decision rule
occurs at a threshold where the fraction of jobs that are engaged
in discrimination more than triples. If the worst case DGP is one
where most jobs are guilty, it makes sense to accuse more jobs. The
lesson here is that although mispecification can lead to substantially
higher risk, ambiguity regarding $G\left(\cdot,\cdot\right)$ will
tend to lead to more rather than fewer audits.

\section{Conclusion}

Correspondence studies are powerful tools that have been extensively
used to detect market level averages of discriminatory behavior. Revisiting
three such studies, we find tremendous heterogeneity across employers
in their degree of discriminatory behavior. This heterogeneity presents
authorities charged with enforcing anti-discrimination laws with a
difficult inferential task. Our analysis suggests that when ensemble
evidence is used, 10 applications per employer is enough to accurately
detect a non-trivial share of discriminatory employers. This finding
opens the possibility that discrimination can be monitored -- perhaps
in real time -- at the employer level.

Our results also provide a number of methodological lessons regarding
the design and analysis of correspondence studies, and of experimental
ensembles more generally. First, we demonstrate that indirect evidence
can serve as a valuable supplement to direct evidence when making
inferences regarding the behavioral responses of particular experimental
units. Our logit results, in particular, suggest that accurately monitoring
illegal discrimination in online labor markets is feasible with relatively
small modifications to conventional audit designs once knowledge of
the callback distribution $G\left(\cdot,\cdot\right)$ has been obtained.
Whether such knowledge is better obtained through sequential experimentation
\citep[e.g.,][]{chakraborty2014dynamic,dimakopoulou2017estimation,narita2018efficient}
or static empirical Bayes methods of the sort considered in this paper
is an interesting question for future work.

Second, our analysis demonstrates that partial identification of the
population distribution of response heterogeneity does not preclude
``borrowing strength'' from experimental ensembles. Using only a
few moments of the callback distribution, we are able to derive informative
lower bounds on the fraction of jobs engaging in illegal discrimination.
These bounds are shown to allow precise inferences to be drawn about
some jobs even in standard designs with only four applications per
job. In the \citet{AGCV} study, which sent eight applications to
each job, we are able to deduce informative lower bound rates of discrimination
against men and women separately.

Third, our results highlight that the appropriate use of indirect
evidence depends critically on the objectives of the investigator,
formalized in our framework by the loss function of a hypothetical
auditor. While in point identified settings it is straightforward
to characterize the tradeoff between type I and II errors implied
by different decision rules, partial identification of heterogeneity
distributions tends to undermine identifiability of this tradeoff
itself, an issue emphasized by \citet{manski2000identification}. In
our setting acknowledging the ambiguity stemming from partial identification
turns out to lead to ``bolder'' inferences, but it is easy to envision
settings where the opposite would be true. An interesting topic for
future research is the extent to which the policy implications of
recent econometric evaluations of teachers, schools, hospitals, and
neighborhoods \citep[e.g.,][]{chetty/friedman/rockoff:14b,ahpw_vam,hull_jmp,chettyhendren_neighborhoods_2,chetty_opportunity_atlas}
vary with alternative notions of risk.

\pagebreak{}

\bibliographystyle{ecta}
\bibliography{cw}

\begin{thebibliography}{55}
\newcommand{\enquote}[1]{``#1''}
\expandafter\ifx\csname natexlab\endcsname\relax\def\natexlab#1{#1}\fi

\bibitem[\protect\citeauthoryear{7th Circuit Court~of Appeals}{7th Circuit
  Court~of Appeals}{2006}]{EEOCvTarget}
\textsc{7th Circuit Court~of Appeals} (2006): \enquote{EEOC v Target Corp.} 460
  (F. 3d), 946.

\bibitem[\protect\citeauthoryear{Altonji and Blank}{Altonji and
  Blank}{1999}]{altonji_blank}
\textsc{Altonji, J.~G. and R.~M. Blank} (1999): \enquote{Race and gender in the
  labor market,} in \emph{Handbook of Labor Economics}, ed. by O.~C.
  Ashenfelter and D.~Card, Elsevier, vol.~3C, chap.~48, 3143--3259.

\bibitem[\protect\citeauthoryear{Angrist}{Angrist}{2014}]{angrist_peers}
\textsc{Angrist, J.~D.} (2014): \enquote{The perils of peer effects,}
  \emph{Labour Economics}, 30, 98--108.

\bibitem[\protect\citeauthoryear{Angrist, Hull, Pathak, and Walters}{Angrist
  et~al.}{2017}]{ahpw_vam}
\textsc{Angrist, J.~D., P.~D. Hull, P.~A. Pathak, and C.~R. Walters} (2017):
  \enquote{Leveraging lotteries for school value-added: testing and
  estimation,} \emph{Quarterly Journal of Economics}, 132, 871--919.

\bibitem[\protect\citeauthoryear{Arceo-Gomez and Campos-Vasquez}{Arceo-Gomez
  and Campos-Vasquez}{2014}]{AGCV}
\textsc{Arceo-Gomez, E.~O. and R.~M. Campos-Vasquez} (2014): \enquote{Race and
  marriage in the labor market: a discrimination correspondence study in a
  developing country,} \emph{American Economic Review: Papers \& Proceedings},
  104, 376--380.

\bibitem[\protect\citeauthoryear{Armstrong}{Armstrong}{2015}]{armstrong2015adaptive}
\textsc{Armstrong, T.} (2015): \enquote{Adaptive testing on a regression
  function at a point,} \emph{The Annals of Statistics}, 43, 2086--2101.

\bibitem[\protect\citeauthoryear{Becker}{Becker}{1957}]{becker_discrimination}
\textsc{Becker, G.~S.} (1957): \emph{The Economics of Discrimination}, The
  University of Chicago Press.

\bibitem[\protect\citeauthoryear{Benjamini and Hochberg}{Benjamini and
  Hochberg}{1995}]{benjamini_hochberg}
\textsc{Benjamini, Y. and Y.~Hochberg} (1995): \enquote{Controlling the false
  discovery rate: a practical and powerful approach to multiple testing,}
  \emph{Journal of the Royal Statistical Society}, 57, 289--300.

\bibitem[\protect\citeauthoryear{Berger}{Berger}{2013}]{berger2013statistical}
\textsc{Berger, J.~O.} (2013): \emph{Statistical decision theory and Bayesian
  analysis}, Springer Science \& Business Media.

\bibitem[\protect\citeauthoryear{Berger}{Berger}{1979}]{berger1979gamma}
\textsc{Berger, R.~L.} (1979): \enquote{Gamma minimax robustness of bayes
  rules: Gamma minimax robustness,} \emph{Communications in Statistics-Theory
  and Methods}, 8, 543--560.

\bibitem[\protect\citeauthoryear{Bertrand and Duflo}{Bertrand and
  Duflo}{2017}]{bertrand_duflo_review}
\textsc{Bertrand, M. and E.~Duflo} (2017): \enquote{Field experiments on
  discrimination,} in \emph{Handbook of Field Experiments}, ed. by E.~Duflo and
  A.~Banerjee, Elsevier, vol.~1.

\bibitem[\protect\citeauthoryear{Bertrand and Mullainathan}{Bertrand and
  Mullainathan}{2004}]{bertrand_mullainathan_2004}
\textsc{Bertrand, M. and S.~Mullainathan} (2004): \enquote{Are Emily and Greg
  more employable than Lakisha and Jamal? A field experiment on labor market
  discrimination,} \emph{American Economic Review}, 94, 991--1013.

\bibitem[\protect\citeauthoryear{Chakraborty and Murphy}{Chakraborty and
  Murphy}{2014}]{chakraborty2014dynamic}
\textsc{Chakraborty, B. and S.~A. Murphy} (2014): \enquote{Dynamic treatment
  regimes,} \emph{Annual review of statistics and its application}, 1,
  447--464.

\bibitem[\protect\citeauthoryear{Charles and Guryan}{Charles and
  Guryan}{2008}]{charles_guryan_jpe}
\textsc{Charles, K.~K. and J.~Guryan} (2008): \enquote{Prejudice and wages: an
  empirical assessment of Becker's The Economics of Discrimination,}
  \emph{Journal of Political Economy}, 116, 773--809.

\bibitem[\protect\citeauthoryear{Chernozhukov, Newey, and Santos}{Chernozhukov
  et~al.}{2015}]{chernozhukov2015constrained}
\textsc{Chernozhukov, V., W.~K. Newey, and A.~Santos} (2015):
  \enquote{Constrained conditional moment restriction models,} \emph{arXiv
  preprint arXiv:1509.06311}.

\bibitem[\protect\citeauthoryear{Chetty, Friedman, Hendren, Jones, and
  Porter}{Chetty et~al.}{2018}]{chetty_opportunity_atlas}
\textsc{Chetty, R., J.~N. Friedman, N.~Hendren, M.~R. Jones, and S.~R. Porter}
  (2018): \enquote{The opportunity atlas: mapping the childhood roots of social
  mobility,} Working paper.

\bibitem[\protect\citeauthoryear{Chetty, Friedman, and Rockoff}{Chetty
  et~al.}{2014{\natexlab{a}}}]{chetty/friedman/rockoff:14a}
\textsc{Chetty, R., J.~N. Friedman, and J.~E. Rockoff} (2014{\natexlab{a}}):
  \enquote{Measuring the impact of teachers I: evaluating bias in teacher
  value-added estimates,} \emph{American Economic Review}, 104, 2593--2563.

\bibitem[\protect\citeauthoryear{Chetty, Friedman, and Rockoff}{Chetty
  et~al.}{2014{\natexlab{b}}}]{chetty/friedman/rockoff:14b}
---\hspace{-.1pt}---\hspace{-.1pt}--- (2014{\natexlab{b}}): \enquote{Measuring
  the impact of teachers II: teacher value-added and student outcomes in
  adulthood,} \emph{American Economic Review}, 104, 2633--2679.

\bibitem[\protect\citeauthoryear{Chetty and Hendren}{Chetty and
  Hendren}{2018}]{chettyhendren_neighborhoods_2}
\textsc{Chetty, R. and N.~Hendren} (2018): \enquote{Impacts of neighborhoods on
  intergenerational mobility II: county-level estimates,} \emph{Quarterly
  Journal of Economics}, 133, 1163--1228, nBER working paper no. 23002.

\bibitem[\protect\citeauthoryear{Chetverikov, Santos, and Shaikh}{Chetverikov
  et~al.}{2018}]{chetverikov2018econometrics}
\textsc{Chetverikov, D., A.~Santos, and A.~M. Shaikh} (2018): \enquote{The
  econometrics of shape restrictions,} \emph{Annual Review of Economics}, 10,
  31--63.

\bibitem[\protect\citeauthoryear{DeGroot}{DeGroot}{2004}]{degroot2005optimal}
\textsc{DeGroot, M.~H.} (2004): \emph{Optimal Statistical Decisions}, vol.~82,
  John Wiley \& Sons.

\bibitem[\protect\citeauthoryear{Dimakopoulou, Athey, and Imbens}{Dimakopoulou
  et~al.}{2017}]{dimakopoulou2017estimation}
\textsc{Dimakopoulou, M., S.~Athey, and G.~Imbens} (2017): \enquote{Estimation
  considerations in contextual bandits,} \emph{arXiv preprint
  arXiv:1711.07077}.

\bibitem[\protect\citeauthoryear{Efron}{Efron}{2004}]{efron2004large}
\textsc{Efron, B.} (2004): \enquote{Large-scale simultaneous hypothesis
  testing: the choice of a null hypothesis,} \emph{Journal of the American
  Statistical Association}, 99, 96--104.

\bibitem[\protect\citeauthoryear{Efron}{Efron}{2010}]{efron_2010}
---\hspace{-.1pt}---\hspace{-.1pt}--- (2010): \enquote{The future of indirect
  evidence,} \emph{Statistical Science}, 25, 145--157.

\bibitem[\protect\citeauthoryear{Efron}{Efron}{2012}]{efron2012large}
---\hspace{-.1pt}---\hspace{-.1pt}--- (2012): \emph{Large-scale inference:
  empirical Bayes methods for estimation, testing, and prediction}, vol.~1,
  Cambridge University Press.

\bibitem[\protect\citeauthoryear{Efron, Tibshirani, Storey, and Tusher}{Efron
  et~al.}{2001}]{efron2001empirical}
\textsc{Efron, B., R.~Tibshirani, J.~D. Storey, and V.~Tusher} (2001):
  \enquote{Empirical Bayes analysis of a microarray experiment,} \emph{Journal
  of the American statistical association}, 96, 1151--1160.

\bibitem[\protect\citeauthoryear{Fang and Santos}{Fang and
  Santos}{2018}]{fang2018inference}
\textsc{Fang, Z. and A.~Santos} (2018): \enquote{Inference on directionally
  differentiable functions,} \emph{The Review of Economic Studies}, 86,
  377--412.

\bibitem[\protect\citeauthoryear{Farber, Silverman, and von Wachter}{Farber
  et~al.}{2016}]{farber_etal_audit}
\textsc{Farber, H.~S., D.~Silverman, and T.~von Wachter} (2016):
  \enquote{Determinants of callbacks to job applications: an audit study,}
  \emph{American Economic Review: Papers \& Proceedings}, 106, 314--318.

\bibitem[\protect\citeauthoryear{Finkelstein, Gentzkow, Hull, and
  Williams}{Finkelstein et~al.}{2017}]{finkelstein_nejm_shrinkage}
\textsc{Finkelstein, A., M.~Gentzkow, P.~Hull, and H.~Williams} (2017):
  \enquote{Adjusting risk adjustment - accounting for variation in diagnostic
  intensity,} \emph{New England Journal of Medicine}, 376, 608--610.

\bibitem[\protect\citeauthoryear{Fisher}{Fisher}{1922}]{fisher_1922}
\textsc{Fisher, R.~A.} (1922): \enquote{On the interpretation of Chi-squared
  from contingency tables, and the calculation of P,} \emph{Journal of the
  Royal Statistical Society}, 85, 87--94.

\bibitem[\protect\citeauthoryear{Fryer and Levitt}{Fryer and
  Levitt}{2004}]{fryer_levitt_2004}
\textsc{Fryer, R.~G. and S.~D. Levitt} (2004): \enquote{The causes and
  consequences of distinctively black names,} \emph{Quarterly Journal of
  Economics}, 119, 767--805.

\bibitem[\protect\citeauthoryear{Guryan and Charles}{Guryan and
  Charles}{2013}]{charles_guryan_jel}
\textsc{Guryan, J. and K.~K. Charles} (2013): \enquote{Taste-based or
  statistical discrimination: the economics of discrimination returns to its
  roots,} \emph{The Economic Journal}, 123, F417--F432.

\bibitem[\protect\citeauthoryear{Heckman, Smith, and Clements}{Heckman
  et~al.}{1997}]{heckman_smith_clements}
\textsc{Heckman, J.~J., J.~Smith, and N.~Clements} (1997): \enquote{Making the
  most out of programme evaluations of social experiments: accounting for
  heterogeneity in programme impacts,} \emph{Review of Economic Studies}, 64,
  487--535.

\bibitem[\protect\citeauthoryear{Hodges, Lehmann et~al.}{Hodges
  et~al.}{1952}]{hodges1952use}
\textsc{Hodges, J.~L., E.~L. Lehmann, et~al.} (1952): \enquote{The use of
  previous experience in reaching statistical decisions,} \emph{The Annals of
  Mathematical Statistics}, 23, 396--407.

\bibitem[\protect\citeauthoryear{Hong and Li}{Hong and
  Li}{2017}]{hong2017numerical}
\textsc{Hong, H. and J.~Li} (2017): \enquote{The numerical delta method and
  bootstrap,} Tech. rep., Working Paper.

\bibitem[\protect\citeauthoryear{Hull}{Hull}{2018}]{hull_jmp}
\textsc{Hull, P.~D.} (2018): \enquote{Estimating hospital quality with
  quasi-experimental data,} Working paper.

\bibitem[\protect\citeauthoryear{Imbens and Rubin}{Imbens and
  Rubin}{2015}]{imbens_rubin_2015}
\textsc{Imbens, G.~W. and D.~B. Rubin} (2015): \emph{Causal inference for
  statistics, social, and medical sciences}, Cambridge University Press.

\bibitem[\protect\citeauthoryear{Kane and Staiger}{Kane and
  Staiger}{2008}]{kane/staiger:08}
\textsc{Kane, T.~J. and D.~O. Staiger} (2008): \enquote{Estimating teacher
  impacts on student achievement: an experimental evaluation,} NBER Working
  Paper 14607.

\bibitem[\protect\citeauthoryear{Kitagawa and Tetenov}{Kitagawa and
  Tetenov}{2018}]{kitagawa2018should}
\textsc{Kitagawa, T. and A.~Tetenov} (2018): \enquote{Who should be treated?
  empirical welfare maximization methods for treatment choice,}
  \emph{Econometrica}, 86, 591--616.

\bibitem[\protect\citeauthoryear{Kleinberg, Ludwig, Mullainathan, and
  Sunstein}{Kleinberg et~al.}{2019}]{kleinberg2019discrimination}
\textsc{Kleinberg, J., J.~Ludwig, S.~Mullainathan, and C.~R. Sunstein} (2019):
  \enquote{Discrimination in the Age of Algorithms,} \emph{Available at SSRN
  3329669}.

\bibitem[\protect\citeauthoryear{Lehmann and Casella}{Lehmann and
  Casella}{2006}]{lehmann2006theory}
\textsc{Lehmann, E.~L. and G.~Casella} (2006): \emph{Theory of point
  estimation}, Springer Science \& Business Media.

\bibitem[\protect\citeauthoryear{Manski}{Manski}{2000}]{manski2000identification}
\textsc{Manski, C.~F.} (2000): \enquote{Identification problems and decisions
  under ambiguity: Empirical analysis of treatment response and normative
  analysis of treatment choice,} \emph{Journal of Econometrics}, 95, 415--442.

\bibitem[\protect\citeauthoryear{M{\"u}ller and Wang}{M{\"u}ller and
  Wang}{2019}]{muller2019nearly}
\textsc{M{\"u}ller, U.~K. and Y.~Wang} (2019): \enquote{Nearly weighted risk
  minimal unbiased estimation,} \emph{Journal of Econometrics}, 209, 18--34.

\bibitem[\protect\citeauthoryear{Narita}{Narita}{2019}]{narita2019experiment}
\textsc{Narita, Y.} (2019): \enquote{Experiment-as-Market: Incorporating
  Welfare into Randomized Controlled Trials,} \emph{Available at SSRN 3094905}.

\bibitem[\protect\citeauthoryear{Narita, Yasui, and Yata}{Narita
  et~al.}{2018}]{narita2018efficient}
\textsc{Narita, Y., S.~Yasui, and K.~Yata} (2018): \enquote{Efficient
  counterfactual learning from bandit feedback,} .

\bibitem[\protect\citeauthoryear{Neyman}{Neyman}{1923}]{neyman_1923}
\textsc{Neyman, J.} (1923): \enquote{On the application of probability theory
  to agricultural experiments,} \emph{Statistical Science}, 5, 465--480.

\bibitem[\protect\citeauthoryear{Noubiap, Seidel et~al.}{Noubiap
  et~al.}{2001}]{noubiap2001algorithm}
\textsc{Noubiap, R.~F., W.~Seidel, et~al.} (2001): \enquote{An algorithm for
  calculating $\Gamma$-minimax decision rules under generalized moment
  conditions,} \emph{The Annals of Statistics}, 29, 1094--1116.

\bibitem[\protect\citeauthoryear{Nunley, Pugh, Romero, and Seals}{Nunley
  et~al.}{2015}]{nunley_etal_2015}
\textsc{Nunley, J.~M., A.~Pugh, N.~Romero, and R.~A. Seals} (2015):
  \enquote{Racial discrimination in the labor market for recent college
  graduates: evidence from a field experiment,} \emph{B.E. Journal of Economic
  Analysis and Policy}, 15, 1093--1125.

\bibitem[\protect\citeauthoryear{Robbins}{Robbins}{1964}]{robbins1964empirical}
\textsc{Robbins, H.} (1964): \enquote{The empirical Bayes approach to
  statistical decision problems,} \emph{The Annals of Mathematical Statistics},
  35, 1--20.

\bibitem[\protect\citeauthoryear{Rubin}{Rubin}{1980}]{rubin_sutva}
\textsc{Rubin, D.~B.} (1980): \enquote{Randomization analysis of experimental
  data: the Fisher Randomization test comment,} \emph{Journal of the American
  Statistical Association}, 75, 591--593.

\bibitem[\protect\citeauthoryear{Savage}{Savage}{1951}]{savage1951theory}
\textsc{Savage, L.~J.} (1951): \enquote{The theory of statistical decision,}
  \emph{Journal of the American Statistical association}, 46, 55--67.

\bibitem[\protect\citeauthoryear{Storey}{Storey}{2002}]{storey2002direct}
\textsc{Storey, J.~D.} (2002): \enquote{A direct approach to false discovery
  rates,} \emph{Journal of the Royal Statistical Society: Series B (Statistical
  Methodology)}, 64, 479--498.

\bibitem[\protect\citeauthoryear{Storey}{Storey}{2003}]{storey2003positive}
---\hspace{-.1pt}---\hspace{-.1pt}--- (2003): \enquote{The positive false
  discovery rate: a Bayesian interpretation and the q-value,} \emph{The Annals
  of Statistics}, 31, 2013--2035.

\bibitem[\protect\citeauthoryear{Tebaldi, Torgovitsky, and Yang}{Tebaldi
  et~al.}{2019}]{tebaldi_torgovitsky_yang}
\textsc{Tebaldi, P., A.~Torgovitsky, and H.~Yang} (2019):
  \enquote{Nonparametric estimates of demand in the California health insurance
  exchange,} Working paper.

\bibitem[\protect\citeauthoryear{Wald}{Wald}{1945}]{wald1945statistical}
\textsc{Wald, A.} (1945): \enquote{Statistical decision functions which
  minimize the maximum risk,} \emph{Annals of Mathematics}, 265--280.

\end{thebibliography}

\pagebreak{}

\appendix
\titleformat{\section}{\normalfont\Large\bfseries}{Appendix \thesection:}{1em}{}

\section{Proof That $pFDR_{J}=\Pr\left(D_{j}=0|\delta(C_{j})=1\right)$\label{sec:pFDR_appx}}

By iterated expectations
\begin{eqnarray*}
\mathbb{E}\left[N_{J}^{-1}\sum_{j=1}^{J}\delta\left(C_{j}\right)\left(1-D_{j}\right)|N_{J}\geq1\right] & = & \sum_{n=1}^{J}n^{-1}\mathbb{E}\left[\sum_{j=1}^{J}\delta\left(C_{j}\right)\left(1-D_{j}\right)|N_{J}=n\right]\Pr\left(N_{J}=n|N_{J}\geq1\right)\\
 & = & \sum_{n=1}^{J}n^{-1}J\mathbb{E}\left[\delta\left(C_{j}\right)\left(1-D_{j}\right)|N_{J}=n\right]\Pr\left(N_{J}=n|N_{J}\geq1\right)\\
 & = & \sum_{n=1}^{J}n^{-1}J\Pr\left(D_{j}=0|\delta(C_{j})=1,N_{J}=n\right)\Pr\left(N_{J}=n|N_{J}\geq1\right)\\
 &  & \times\Pr\left(\delta(C_{j})=1|N_{J}=n\right)\\
 & = & \sum_{n=1}^{J}\Pr\left(D_{j}=0|\delta(C_{j})=1,N_{J}=n\right)\Pr\left(N_{J}=n|N_{J}\geq1\right)\\
 & = & \Pr\left(D_{j}=0|\delta(C_{j})=1\right)\sum_{n=1}^{J}\Pr\left(N_{J}=n|N_{J}\geq1\right)\\
 & = & \Pr\left(D_{j}=0|\delta(C_{j})=1\right),
\end{eqnarray*}
where the second and fifth lines use that the $\left\{ C_{j},D_{j}\right\} _{j=1}^{J}$
are $iid$ and the fourth uses the fact that $\Pr\left(\delta\left(C_{j}\right)=1|N_{J}=n\right)=n/J$.
Hence, $pFDR_{J}$ gives the probability $\Pr\left(D_{j}=0|\delta(C_{j})=1\right)$
that an investigated job is innocent.

\section{Discretization of $G$ and Linear Programming Bounds\label{sec:LP_appx}}

To compute the solution to the problem in (\ref{eq:LP}), we approximate
the CDF $G\left(p_{w},p_{b}\right)$ with the discrete distribution
\begin{eqnarray*}
G_{K}\left(p_{w},p_{b}\right) & = & \sum_{k=1}^{K}\sum_{l=1}^{K}\pi_{kl}1\left\{ p_{w}\leq\varrho\left(k,l\right),p_{b}\leq\varrho\left(l,k\right)\right\} ,
\end{eqnarray*}
where the $\left\{ \pi_{kl}\right\} _{k=1,l=1}^{K,K}$ are probability
masses and $\left\{ \varrho\left(k,l\right),\varrho\left(l,k\right)\right\} $$_{k=1,l=1}^{K,K}$
comprise a set of mass point coordinates generated by the function
\[
\varrho\left(x,y\right)=\underbrace{\frac{\min\left\{ x,y\right\} -1}{K}}_{\text{diagonal}}+\underbrace{\frac{\max\left\{ 0,x-y\right\} ^{2}}{K\left(1+K-y\right)}}_{\text{off-diagonal}}.
\]
This discretization scheme can be visualized as a two-dimensional
grid containing $K^{2}$ elements. The diagonal entries on the grid
represent jobs where no discrimination is present. The first term
above ensures the mass points are equally spaced along the diagonal
from $\left(0,0\right)$ to $\left(\frac{K-1}{K},\frac{K-1}{K}\right)$.
The second term spaces off diagonal points quadratically according
to their distance from the diagonal in order to accomodate jobs with
very low levels of discrimination while economizing on the number
of grid points. Note that $\underset{K\rightarrow\infty}{\lim}\varrho\left(K,y\right)=1$
ensuring the grid asymptotically spans the unit square.

With this notation, the constraints in (\ref{eq:LPcons}) can be written:
\begin{equation}
\bar{f}\left(c_{w},c_{b}\right)=\left(\begin{array}{c}
L_{w}\\
c_{w}
\end{array}\right)\left(\begin{array}{c}
L_{b}\\
c_{b}
\end{array}\right)\sum_{k=1}^{K}\sum_{l=1}^{K}\pi_{kl}\varrho\left(k,l\right)^{c_{w}}\left(1-\varrho\left(k,l\right)\right)^{L_{w}-c_{w}}\varrho\left(l,k\right)^{c_{b}}\left(1-\varrho\left(l,k\right)\right)^{L_{b}-c_{b}},\label{eq:discrete_cons}
\end{equation}
for $c_{w}=\left(1,...,L_{w}\right)$ and $c_{b}=\left(1,...,L_{b}\right)$.
Hence, our composite discretized optimization problem is to
\[
\max_{\left\{ \pi_{kl}\right\} }\frac{\left(\begin{array}{c}
L\\
t
\end{array}\right)}{\sum_{\left(c_{w}',c_{b}'\right):c_{w}'+c_{b}'=t}\bar{f}\left(c_{w}',c_{b}'\right)}\sum_{l=0}^{K}\sum_{k=0}^{K}\pi_{kl}\varrho\left(k,l\right)^{t}\left(1-\varrho\left(k,l\right)\right)^{L-t},
\]
subject to (\ref{eq:discrete_cons}) and
\[
\sum_{k=1}^{K}\sum_{l=1}^{K}\pi_{kl}=1,\quad\pi_{kl}\geq0,
\]
for $k=1,...,K$ and $l=1,...,K$. We solve this problem numerically
using the Gurobi software package. Because setting $K$ too low will
tend to yield artificially tight bounds, we set $K=900$ in all bound
computation steps, which yields $\left(900\right)^{2}=810,000$ distinct
mass points.

\section{Shape Constrained GMM\label{sec:shapeconstrained_appx}}

To accomodate the \citet{nunley_etal_2015} study which employs multiple
application designs, we introduce the variable $A_{j}=\left(A_{jw},A_{jb}\right)$
which gives the number of white and black applications sent to job
$j$. Collecting the design-specific callback probabilities $\left\{ \Pr\left(C_{jw}=c_{w},C_{jb}=c_{b}|A_{j}=a\right)\right\} _{c_{w},c_{b}}$
into the vector $f_{a}$, our model relates these probabilities to
moments of the callback distribution via the linear system $f_{a}=B_{a}\mu$,
for $B_{a}$ a fixed matrix of binomial coefficients. Letting $f$
denote the vector formed by ``stacking'' the $\left\{ f_{a}\right\} $
across designs in an experiment, we write $f=B\mu.$ Let $\pi$ be
a $K^{2}\times1$ vector comprised of the probability masses $\left\{ \pi_{kl}\right\} _{k=1,l=1}^{K,K}$
(see Appendix \ref{sec:LP_appx}). For GMM estimation we set $K=150$
(larger values yield very similar results). From (\ref{eq:discrete_cons}),
we can write $\mu=M\pi$ where $M$ is a $dim\left(\mu\right)\times K^{2}$
matrix comprised of entries with typical element $\varrho\left(k,l\right)^{m}\left(1-\varrho\left(k,l\right)\right)^{s-m}\varrho\left(l,k\right)^{n}\left(1-\varrho\left(l,k\right)\right)^{t-n}$.
Defining $R=BM$, we have the moment restriction $f=R\pi$.

Let $\tilde{f}$ denote the vector of empirical call back probabilities
with typical element:
\begin{center}
$\frac{J^{-1}\sum_{j=1}^{J}1\left\{ C_{jw}=c_{w},C_{jb}=c_{b},A_{j}=a\right\} }{J^{-1}\sum_{j=1}^{J}1\left\{ A_{j}=a\right\} }$.
\par\end{center}

\noindent \begin{flushleft}
Our shape constrained GMM estimator of $\pi$ can be written as the
solution to the following quadratic programming problem:
\begin{equation}
\hat{\pi}=\arg\inf_{\pi}\;(\tilde{f}-R\pi)'W(\tilde{f}-R\pi)\label{eq:GMM}
\end{equation}
\[
\textrm{ s.t. }\;\pi\geq0,\;\textbf{1}'\pi=1,
\]
where $W$ is a fixed weighting matrix. Note that because $G\left(\cdot,\cdot\right)$
is not identified, there are many possible solutions $\hat{\pi}$
to this problem, but these solutions will all yield the same values
of $R\hat{\pi}$. Our shape constrained estimate of the moments is
$\hat{\mu}=M\hat{\pi}$ while our estimator of the callback probabilities
is $\hat{f}=R\hat{\pi}$. We follow a two-step procedure, solving
(\ref{eq:GMM}) with diagonal weights proportional to the number of
jobs used in the application design and then choosing $W=\hat{\Sigma}^{-1}$
where $\hat{\Sigma}=\text{diag}\left(\hat{f}^{\left(1\right)}\right)-\hat{f}^{\left(1\right)}\hat{f}^{\left(1\right)\prime}$
is an estimate of the variance-covariance matrix of the callback frequencies
implied by the first step shape-constrained callback probability estimates
$\hat{f}^{\left(1\right)}$.
\par\end{flushleft}

\subsection*{\citet{hong2017numerical} standard errors}

Standard errors on the moment estimates $\hat{\mu}$ are computed
via the numerical bootstrap procedure of \citet{hong2017numerical}
using a step size of $J^{-1/4}$ (we found qualitatively similar results
with a step size of $J^{-1/3}$). Our implementation of the numerical
bootstrap proceeds as follows: the bootstrap analogue $\mu^{*}$ of
$\hat{\mu}$ solves the quadratic programming problem in (\ref{eq:GMM})
where $\tilde{f}$ has been replaced by $\left(\tilde{f}+J^{-1/4}f^{*}\right)$.
The bootstrap probabilities $f^{*}$ have typical element:
\[
\tfrac{J^{-1}\sum_{j=1}^{J}\omega_{j}^{*}1\left\{ C_{jw}=c_{w},C_{jb}=c_{b},A_{j}=a\right\} }{J^{-1}\sum_{j=1}^{J}\omega_{j}^{*}1\left\{ A_{j}=a\right\} },
\]
where $\left\{ \omega_{j}^{*}\right\} _{j=1}^{J}$ are a set of iid
draws from an exponential distribution with mean and variance one.
For any function $\phi\left(\hat{\mu}\right)$ of the moment estimates
$\hat{\mu}$ reported, we use as our standard error estimate the standard
deviation across bootstrap replications of $J^{-1/4}\left[\phi\left(\mu^{*}\right)-\phi\left(\hat{\mu}\right)\right]$.

\subsection*{\citet{chernozhukov2015constrained} goodness of fit test}

To formally test whether there exists a $\pi$ in the $K^{2}$ dimensional
probability simplex such that $f=R\pi$ holds, we rely on the procedure
of \citet{chernozhukov2015constrained}. Our test statistic (the ``\emph{J}-test'')
can be written:
\[
T_{n}=\inf_{\pi}\;(\tilde{f}-R\pi)'\hat{\Sigma}^{-1}(\tilde{f}-R\pi)
\]

\[
\textrm{ s.t. }\;\pi\geq0,\;\textbf{1}'\pi=1.
\]

Letting $\mathbb{F}^{*}=f^{*}-\tilde{f}$ denote the (centered) bootstrap
analogue of the callback frequencies $\tilde{f}$ and $W^{*}$ a corresponding
bootstrap weighting matrix, our bootstrap test statistic takes the
form: 
\begin{equation}
T_{n}^{*}=\inf_{\pi,h}\;(\mathbb{F}^{*}-Rh)'W^{*}(\mathbb{F}^{*}-Rh)\label{eq:BS_J}
\end{equation}
\[
\textrm{ s.t. }\;(\tilde{f}-R\pi)'W(\tilde{f}-R\pi)=T_{n},\;\pi\geq0,\;\textbf{1}'\pi=1,\;h\geq-\pi,\;\textbf{1}'h=0
\]
As in the full sample problem, we conduct a two-step GMM procedure
in each bootstrap replication, setting $W^{*}=\left[\textrm{diag}(R\pi^{\left(1\right)*})-(R\pi^{\left(1\right)*})(R\pi^{\left(1\right)*})^{\prime}\right]^{-1}$
where $\pi^{\left(1\right)*}$ is a first-step diagonally weighted
estimate of the probabilities in the bootstrap sample. The goodness
of fit \emph{p}-value reported is the fraction of bootstrap samples
for which $T_{n}^{*}>T_{n}$.

To simplify computation of (\ref{eq:BS_J}), we re-formulate the problem
in two ways. First, we define primary and auxilliary vectors of errors
for each moment condition. Letting $\xi_{h}=\mathbb{F}^{*}-Rh$ and
$\xi_{\pi}=\tilde{f}-R\pi$, the problem can be re-posed as: 
\[
T_{n}^{*}=\inf_{\xi_{h},\xi_{\pi}}\;\xi_{h}'W^{*}\xi_{h},
\]
\[
\textrm{s.t.}\;\;\xi_{\pi}'W\xi_{\pi}=T_{n},\;\;Rh+\xi_{h}=\mathbb{F}^{*},\;\;R\pi+\xi_{\pi}=\tilde{f},\;\;\textbf{1}'h=0,\;\;\textbf{1}'\pi=1,\;\;h\geq-\pi,\;\;\pi\geq0.
\]
Now letting $h^{+}=h+\pi$, we can further rewrite the problem as:
\[
T_{n}^{*}=\inf_{\xi_{h},\xi_{\pi}}\;\xi_{h}'W^{*}\xi_{h},
\]
\[
\textrm{s.t.}\;\;\xi_{\pi}'W\xi_{\pi}=T_{n},\;\;Rh^{+}+\xi_{h}+\xi_{\pi}=\mathbb{F}^{*},\;\;R\pi+\xi_{\pi}=\tilde{f},\;\;\textbf{1}'h^{+}=1,\;\;\textbf{1}'\pi=1,\;\;h^{+}\geq0,\;\;\pi\geq0.
\]
Note that this final representation replaces a $K^{2}\times K^{2}+1$
(inequality) constraint matrix encoding $\xi_{h}\geq-\xi_{\pi}$ and
$\xi_{\pi}\geq0$ with a $2K^{2}\times1$ vector encoding $h^{+}\geq0$
and $\pi\geq0$. Because this problem still involves a quadratic constraint
in $\xi_{\pi}$, we make use of Gurobi's Second Order Cone Programming
(SOCP) solver to obtain a solution.

\section{Computing Maximum Risk\label{sec:maxrisk_appx}}

We approximate $G\left(p_{w}^{H},p_{w}^{L},p_{b}^{H},p_{b}^{L}\right)$
with the discretized distribution
\begin{eqnarray*}
G_{K}\left(p_{w}^{H},p_{w}^{L},p_{b}^{H},p_{b}^{L}\right) & = & \sum_{k=1}^{K}\sum_{l=1}^{K}\sum_{k'=1}^{K}\sum_{l'=1}^{K}\pi_{klk'l'}1\left\{ p_{w}^{H}\leq\varrho\left(k,l\right),p_{w}^{L}\leq\varrho\left(k',l'\right),p_{b}^{H}\leq\varrho\left(l,k\right),p_{b}^{L}\leq\varrho\left(l',k'\right)\right\} ,
\end{eqnarray*}
which has $K^{4}$ mass points. In practice, we choose $K=30$, which
yields the same number of points as the approximation described in
Appendix \ref{sec:LP_appx}.

Generalizing the notation of Appendix \ref{sec:shapeconstrained_appx},
let the vector $A_{j}=\left(A_{jw}^{H},A_{jw}^{L},A_{jb}^{H},A_{jb}^{L}\right)$
record the number of high quality and low quality applications of
each race sent to job $j$ and let $C_{j}=\left(C_{jw}^{H},C_{jw}^{L},C_{jb}^{H},C_{jb}^{L}\right)$
record the corresponding numbers of callbacks. The posterior probability
of discrimination is $\Pr\left(D_{j}=1|A_{j},C_{j}\right)=\mathcal{P}\left(C_{j},A_{j},G\right)$.
The space of auditing rules we consider is of the form $\delta\left(C_{j},A_{j},q\right)=1\left\{ \mathcal{P}\left(C_{j},A_{j},G_{\text{logit}}\right)>q\right\} $.With
this notation, we can write the risk function
\begin{eqnarray*}
\mathcal{R}_{J}(q) & = & \sum_{j=1}^{J}\Pr\left(\delta\left(C_{j},A_{j},q\right)=1,D_{j}=0\right)\kappa+\Pr\left(\delta\left(C_{j},A_{j},q\right)=0,D_{j}=1\right)\gamma\\
 & = & J\times\sum_{a\in\mathscr{A}_{1}}w_{a}\left\{ \Pr\left(\delta\left(C_{j},a,q\right)=1,D_{j}=0\right)\kappa+\Pr\left(\delta\left(C_{j},a,q\right)=0,D_{j}=1\right)\gamma\right\} .
\end{eqnarray*}
where $\mathscr{A}_{1}$ is the set of all $2^{5}=36$ binary quality
permutations possible in a design with 5 white and 5 black applications
and $w_{a}=\left(\begin{array}{c}
5\\
a_{w}^{H}
\end{array}\right)\left(\begin{array}{c}
5\\
a_{b}^{H}
\end{array}\right)\left(\frac{1}{2}\right)^{10}$ is the set of weights that arise when quality is assigned at random
within race.

To further evaluate the above risk expression note that when $D_{j}=1-1\left\{ p_{jb}^{H}\geq p_{jw}^{H}\right\} 1\left\{ p_{jb}^{L}\geq p_{jw}^{L}\right\} $,
we can write:
\begin{eqnarray*}
\Pr\left(\delta\left(C_{j},a,q\right)=1,D_{j}=0\right) & = & \sum_{c_{w}^{H}=0}^{a_{jw}^{H}}\sum_{c_{b}^{H}=0}^{a_{jb}^{H}}\sum_{c_{w}^{L}=0}^{a_{jw}^{L}}\sum_{c_{b}^{L}=0}^{a_{jb}^{L}}\delta\left(c,a,q\right)\\
 & \times & \sum_{k=1}^{K}\sum_{l=1}^{K}\sum_{k'\geq k}^{K}\sum_{l'\geq l}^{K}\pi_{klk'l'}\times\left(\begin{array}{c}
a_{w}^{H}\\
c_{w}^{H}
\end{array}\right)\left(\begin{array}{c}
a_{b}^{H}\\
c_{b}^{H}
\end{array}\right)\left(\begin{array}{c}
a_{w}^{L}\\
c_{w}^{L}
\end{array}\right)\left(\begin{array}{c}
a_{b}^{L}\\
c_{b}^{L}
\end{array}\right)\\
 & \times & \varrho\left(k,l\right)^{c_{w}^{H}}\left(1-\varrho\left(k,l\right)\right)^{a_{w}^{H}-c_{w}^{H}}\varrho\left(l,k\right)^{c_{b}^{H}}\left(1-\varrho\left(l,k\right)\right)^{a_{b}^{H}-c_{b}^{H}}\\
 & \times & \varrho\left(k',l'\right)^{c_{w}^{L}}\left(1-\varrho\left(k',l'\right)\right)^{a_{w}^{L}-c_{w}^{L}}\varrho\left(l',k'\right)^{c_{b}^{L}}\left(1-\varrho\left(l',k'\right)\right)^{a_{b}^{L}-c_{b}^{L}}.
\end{eqnarray*}
\begin{eqnarray*}
\Pr\left(\delta\left(C_{j},a,q\right)=0,D_{j}=1\right) & = & \sum_{c_{w}^{H}=0}^{a_{jw}^{H}}\sum_{c_{b}^{H}=0}^{a_{jb}^{H}}\sum_{c_{w}^{L}=0}^{a_{jw}^{L}}\sum_{c_{b}^{L}=0}^{a_{jb}^{L}}\left(1-\delta\left(c,a,q\right)\right)\\
 & \times & \sum_{k=1}^{K}\sum_{l=1}^{K}\sum_{k'<k}^{K}\sum_{l'<l}^{K}\pi_{klk'l'}\times\left(\begin{array}{c}
a_{w}^{H}\\
c_{w}^{H}
\end{array}\right)\left(\begin{array}{c}
a_{b}^{H}\\
c_{b}^{H}
\end{array}\right)\left(\begin{array}{c}
a_{w}^{L}\\
c_{w}^{L}
\end{array}\right)\left(\begin{array}{c}
a_{b}^{L}\\
c_{b}^{L}
\end{array}\right)\\
 & \times & \varrho\left(k,l\right)^{c_{w}^{H}}\left(1-\varrho\left(k,l\right)\right)^{a_{w}^{H}-c_{w}^{H}}\varrho\left(l,k\right)^{c_{b}^{H}}\left(1-\varrho\left(l,k\right)\right)^{a_{b}^{H}-c_{b}^{H}}\\
 & \times & \varrho\left(k',l'\right)^{c_{w}^{L}}\left(1-\varrho\left(k',l'\right)\right)^{a_{w}^{L}-c_{w}^{L}}\varrho\left(l',k'\right)^{c_{b}^{L}}\left(1-\varrho\left(l',k'\right)\right)^{a_{b}^{L}-c_{b}^{L}}.
\end{eqnarray*}

Using these expressions, maximal risk can therefore be written as
the solution to the following linear programming problem taking the
form:
\begin{eqnarray*}
\mathcal{R}_{J}^{m}(q) & = & J\times\max_{\left\{ \pi_{klk'l'}\right\} }\sum_{a\in\mathscr{A}_{1}}w_{a}\left\{ \Pr\left(\delta\left(C_{j},a,q\right)=1,D_{j}=0\right)\kappa+\Pr\left(\delta\left(C_{j},a,q\right)=0,D_{j}=1\right)\gamma\right\} 
\end{eqnarray*}
subject to the constraint that the $\pi_{klk'l'}$ are non-negative
and sum to one and that the following moment restrictions hold:
\begin{eqnarray*}
\Pr\left(C_{j}=c|A_{j}=a\right) & = & \left(\begin{array}{c}
a_{w}^{H}\\
c_{w}^{H}
\end{array}\right)\left(\begin{array}{c}
a_{b}^{H}\\
c_{b}^{H}
\end{array}\right)\left(\begin{array}{c}
a_{w}^{L}\\
c_{w}^{L}
\end{array}\right)\left(\begin{array}{c}
a_{b}^{L}\\
c_{b}^{L}
\end{array}\right)\sum_{k=1}^{K}\sum_{l=1}^{K}\sum_{k'=1}^{K}\sum_{l'=1}^{K}\pi_{klk'l'}\\
 & \times & \varrho\left(k,l\right)^{c_{w}^{H}}\left(1-\varrho\left(k,l\right)\right)^{a_{w}^{H}-c_{w}^{H}}\varrho\left(l,k\right)^{c_{b}^{H}}\left(1-\varrho\left(l,k\right)\right)^{a_{b}^{H}-c_{b}^{H}}\\
 & \times & \varrho\left(k',l'\right)^{c_{w}^{L}}\left(1-\varrho\left(k',l'\right)\right)^{a_{w}^{L}-c_{w}^{L}}\varrho\left(l',k'\right)^{c_{b}^{L}}\left(1-\varrho\left(l',k'\right)\right)^{a_{b}^{L}-c_{b}^{L}}.
\end{eqnarray*}
Specifically, we impose these restrictions for the following set of
designs, all of which are present in the \citet{nunley_etal_2015}
experiment:
\begin{center}
$\mathscr{A}_{2}=\left\{ \left(2,0,2,0\right),\left(2,0,0,2\right),\left(0,2,2,0\right),\left(0,2,0,2\right)\right\} .$
\par\end{center}

\noindent To operationalize these constraints, we replace the unknown
cell probabilities $\Pr\left(C_{j}=c|A_{j}=a\right)$ for all $c$
and $a$ in $\mathscr{A}_{2}$ with their predictions under the logit
model reported in column 2 of Table X. Using the logit predictions
serves as a form of smoothing that allows us to avoid problems that
arise with small cells when considering quality variation due to covariates.\pagebreak{}



\section*{Supplementary material}

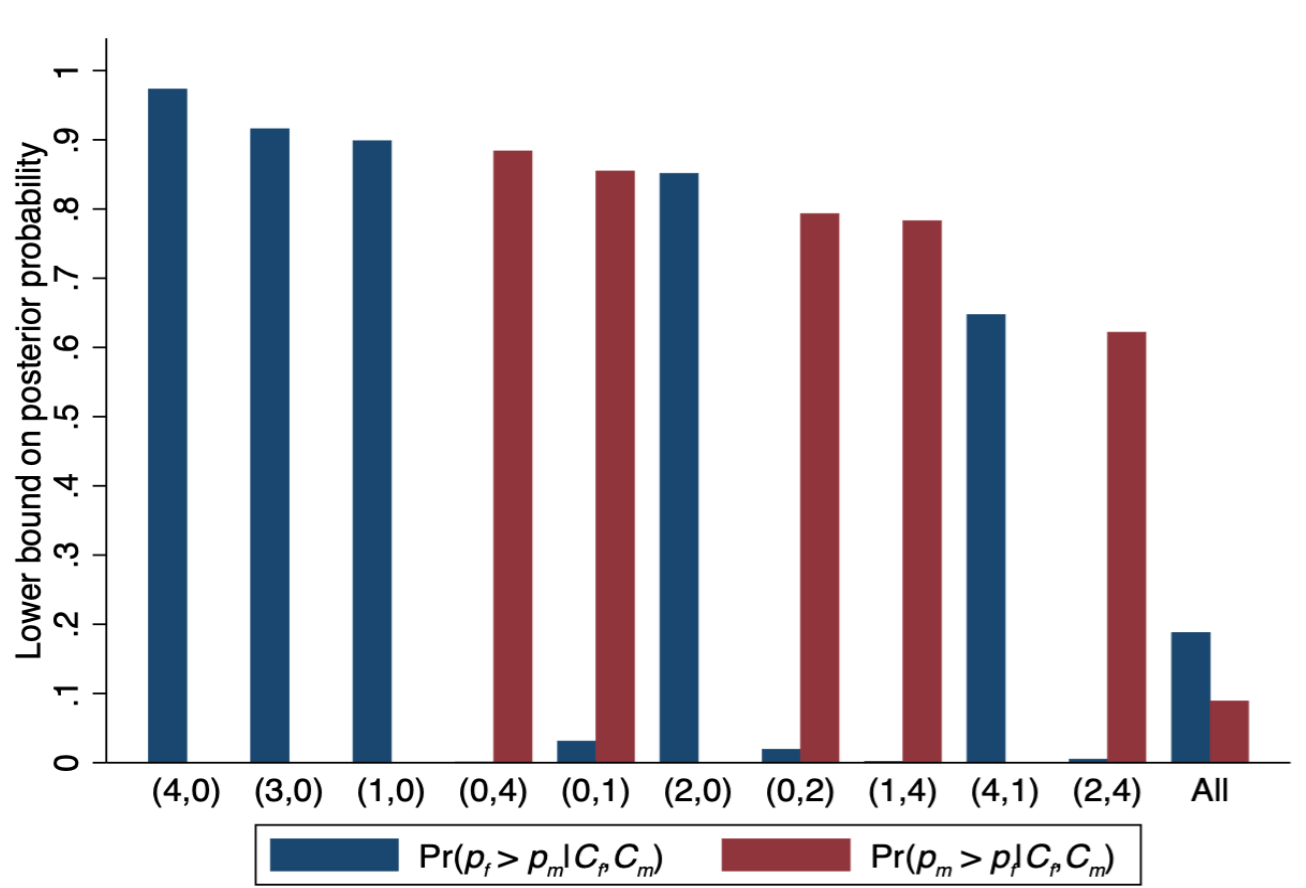

Figure III: Lower bounds on posterior probabilities of discrimination, AGCV data

Notes: This figure displays lower bounds on the probability that jobs in the Arceo-Gomez and Campos-Vasquez (2014) data discriminate based upon sex for the 10 callback configurations with highest posterior bounds. Each bar reports a lower bound on the posterior probability of discrimination conditional on $(C_f, C_m)$, where $C_f$ is the number of female callbacks and $C_m$ is the number of male callbacks. Blue bars report lower bounds on the probability of discriminating against men, and red bars report lower bounds on the probability of discriminating against women.

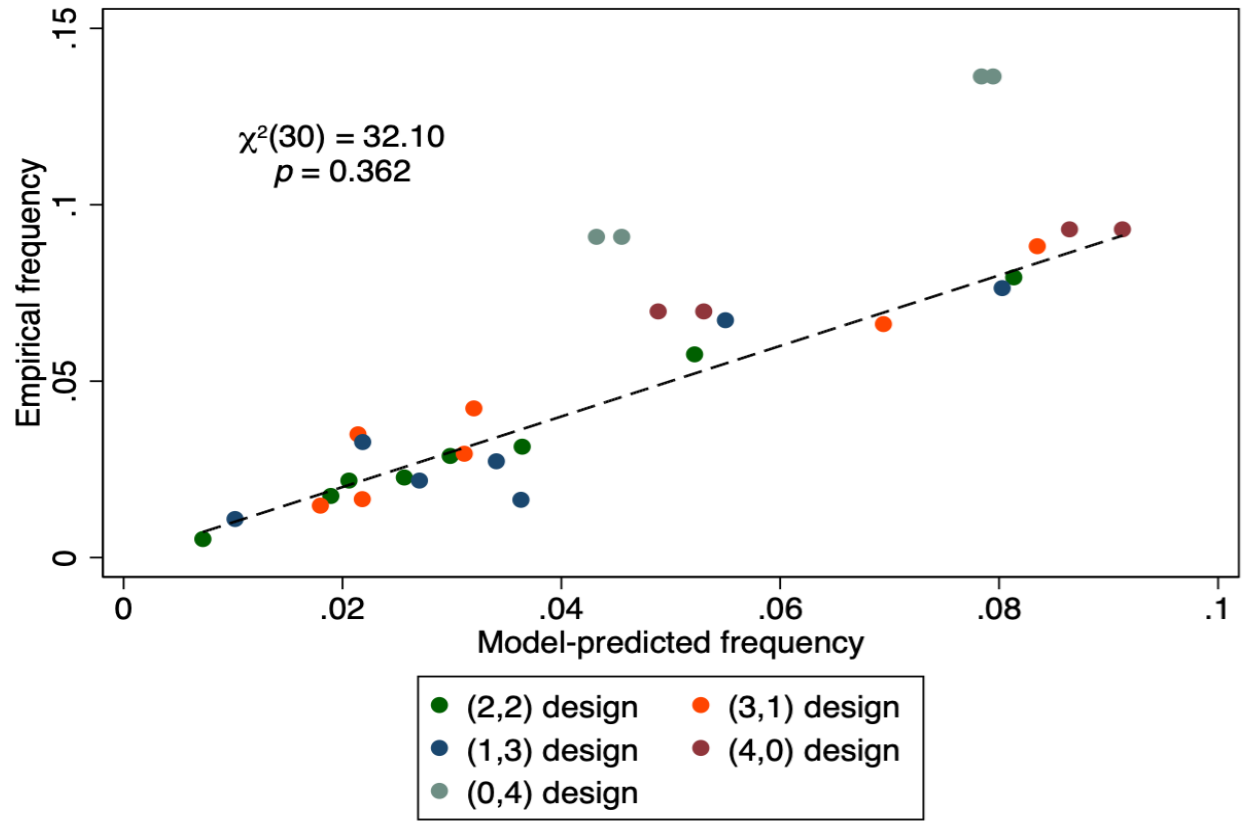

Figure IV: Mixed logit model fit

Notes: This figure compares mixed logit predicted frequencies for callback events in the Nunley et al. (2015) data with corresponding empirical frequencies. The horizontal axis plots model-predicted probabilities for each possible combination of white and black callback counts (excluding zero total callbacks), separately by experimental design. Model predictions are calculated by simulating the logit model in column (2) of Table X 10,000 times for each job in the Nunley et al. data set. The vertical axis plots the observed frequency of each event. Green dots show frequencies for a design with two white and two black applications, while orange, blue, red, and grey points show frequencies for designs with 3 white and 1 black, 1 white and three black, 4 white and zero black, and 0 white and 4 black applications, respectively. The dashed line is the 45-degree line. The chi-squared statistic and $p$-value come from a test that all model-predicted and empirical frequencies match, treating the model predictions as fixed.

Figure V: Mixed logit estimates of posterior discrimination probabilities, Nunley et al. data

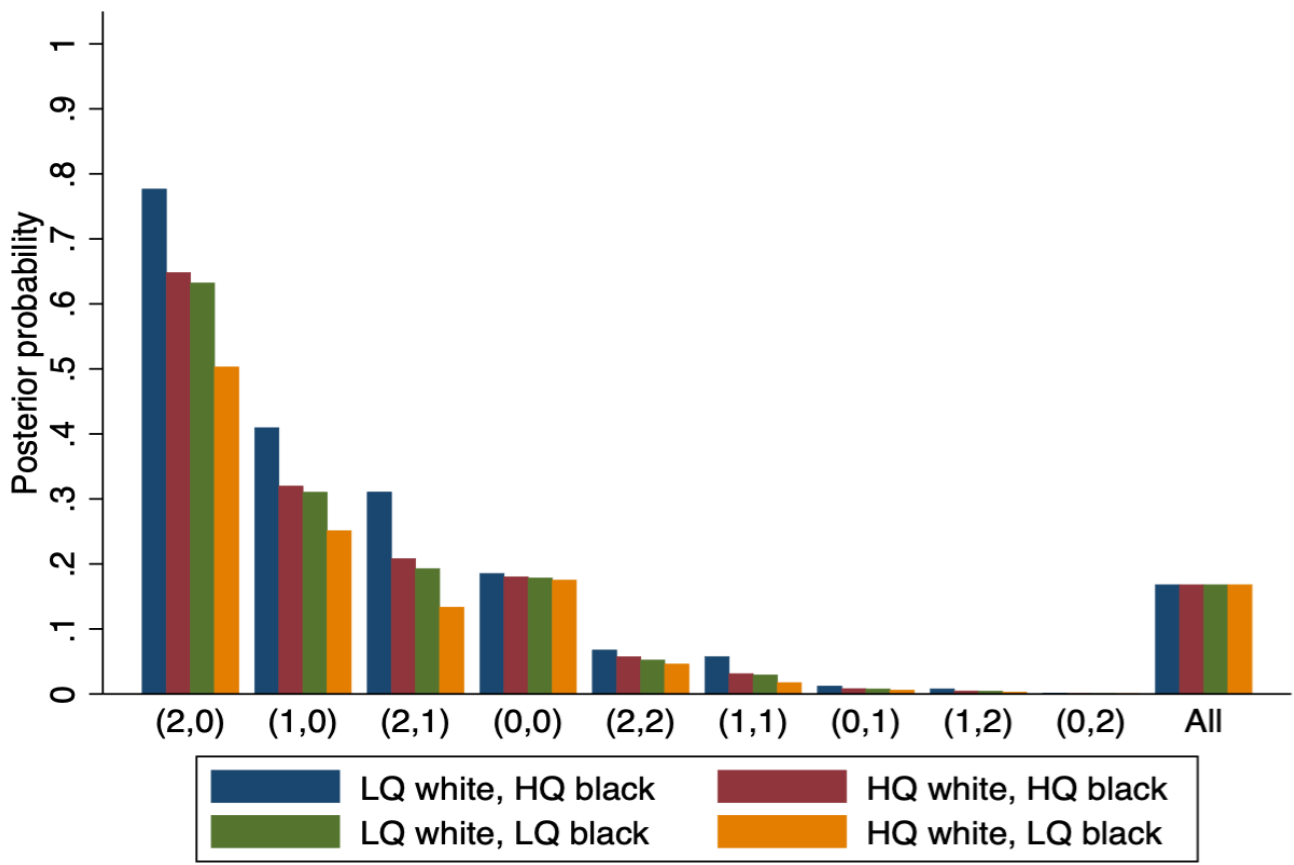

Notes: This figure displays mixed logit estimates of the posterior probability that jobs in the Nunley et al. (2015) data discriminate against black workers conditional on ($C_w$, $C_b$), where $C_w$ is the number of white callbacks and $C_b$ is the number of black callbacks. Blue bars show posteriors for a design sending two low quality (LQ) white applications and two high quality (HQ) black applications, where low and high quality are defined based on a logit covariate index 1 standard deviation below or above the mean. Red bars show posteriors for a design sending two HQ white and two HQ black applications. Green bars show posteriors for a design sending two LQ white and two LQ black applications. Orange bars show posteriors for a design sending two HQ white and two LQ black applications.

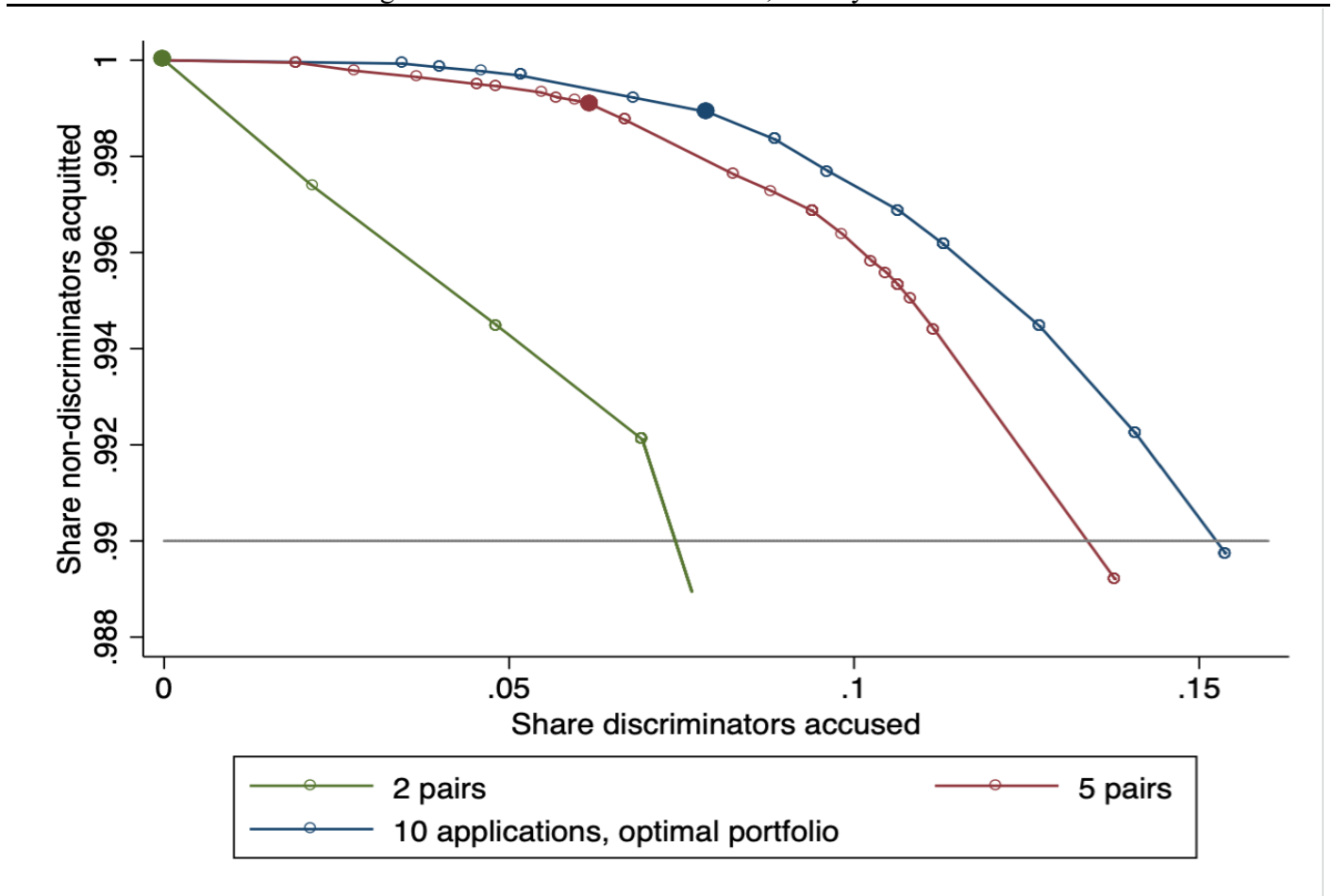

Figure VI: Detection/error tradeoffs, Nunley et al. data

Notes: This figure displays detection/error tradeoff curves based on models fit to the Nunley et al. (2015) data. Estimates come from decision rules applied to experiments generated from the logit model in column (2) of Table X. The horizontal axis measures the share of discriminating jobs accused by each decision rule, while the vertical axis measures the share of non-discriminating jobs not accused (acquitted). The curves are generated by varying the posterior threshold at which firms are accused. The green curve corresponds to an experiment that sends two white and two black applications to each firm, and the red curve corresponds to sending five applications of each race. These two curves randomly assign a 2-valued covariate index of resume quality (high or low), defined as +/-1 the empirical standard deviation of this index. The blue curve shows results from sending the optimal mix of race and resume quality for each posterior threshold. Bold points correspond to 80% posterior thresholds.

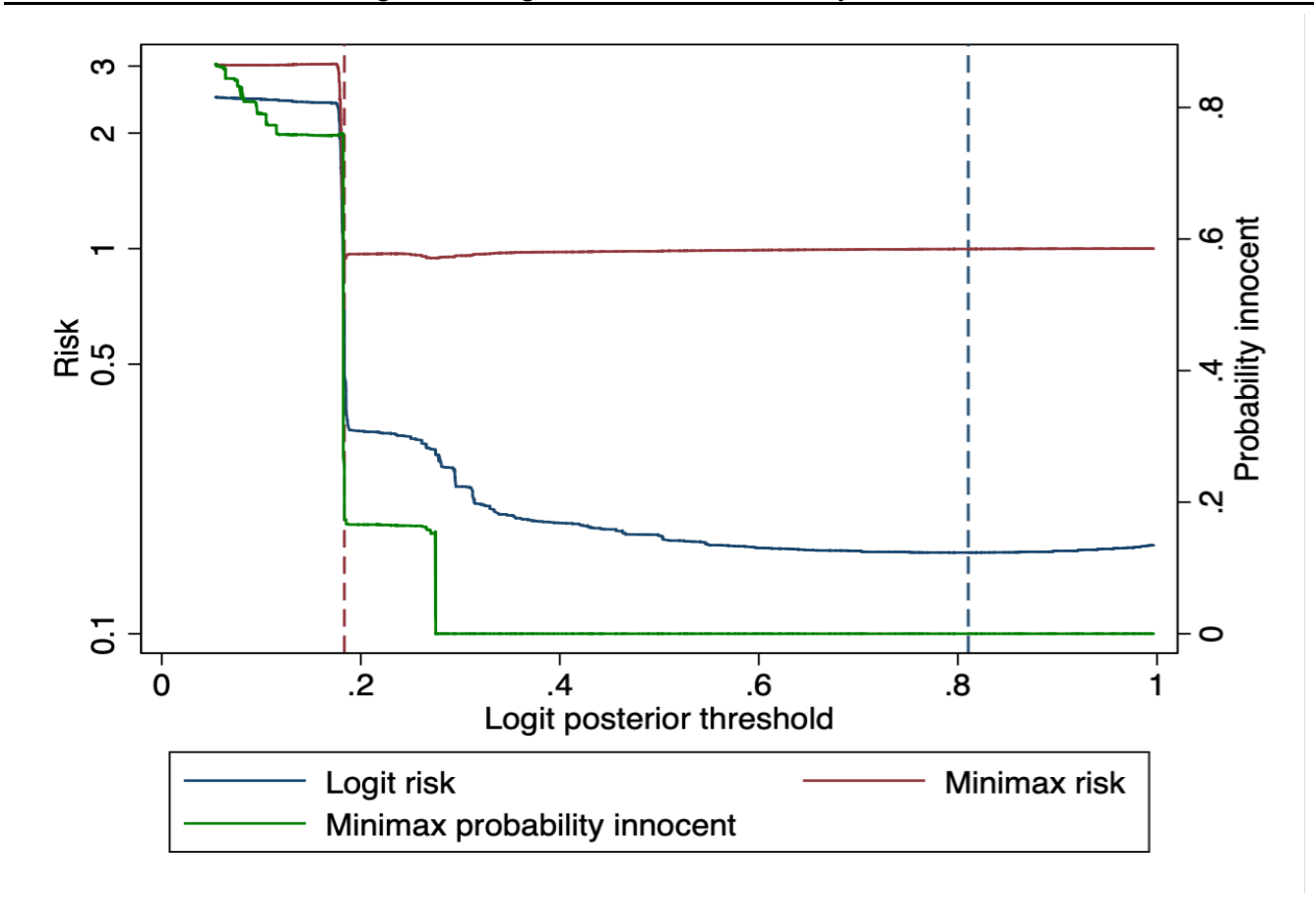

Figure VII: Logit and minimax risk, Nunley et al. data

Notes: This figure displays average risk generated by a hypothetical experiment sending five white and five black applications to jobs in the population studied by Nunley et al. (2015). Resumes are randomly assigned to high or low quality with equal probability, where quality is defined as +/-1 the empirical standard deviation of the logit covariate index. The horizontal axis plots the posterior threshold at which jobs are accused of discrimination. The risk function is 4 times the number of non-discriminators accused plus the number of discriminators not accused. Discrimination is defined as calling white applicants at a higher rate than black applicants. The blue curve displays risk using the logit data generating process in column (2) of Table X. The red curve plots minimax risk calculated by choosing the joint distribution of callback probabilities to maximize risk for each decision rule subject to the moments identified by the Nunley et al. (2015) experiment, restricted to decision rules that order jobs by the logit posterior threshold. Vertical dashed lines indicate risk-minimizing thresholds. The green curve displays the share of jobs that are innocent in the worst-case data generating process.

Table I: Descriptive statistics for resume correspondence studies

|  | Bertrand & Mullainathan (1) | Nunley et al. (2) | Arceo-Gomez & Campos-Vasquez (3) |
|---|---|---|---|
| Number of jobs | 1,112 | 2,305 | 802 |
| Applications per job | 4 | 4 | 8 |
| Treatment/control | Black/white | Black/white | Male/female |
| Callback rates: Total | 0.079 | 0.167 | 0.123 |
| Treatment | 0.063 | 0.154 | 0.108 |
| Control | 0.094 | 0.180 | 0.138 |
| Difference | -0.031 (0.007) | -0.026 (0.007) | -0.030 (0.008) |

Notes: This table reports sample characteristics based on data from three resume correspondence experiments. Columns (1) and (2) show statistics from Bertrand and Mullainathan's (2004) and Nunley et al.'s (2015) studies of racial discrimination in the United States. Column (3) reports statistics from Arceo-Gomez & Campos-Vasquez's (2014) study of gender discrimination in Mexico. Standard errors for treatment/control differences, clustered at the job level, are in parentheses.

Table II: Tests for dependence across trials

| | Nunley et al. data | | | AGCV data | |
|---|---|---|---|---|---|
| | Main effect | Leave-out mean | | Main effect | Leave-out mean |
| Variable | (1) | (2) | Variable | (3) | (4) |
| Black | -0.028 | -0.019 | Married | 0.001 | 0.002 |
| | (0.010) | (0.027) | | (0.008) | (0.033) |
| Female | 0.010 | 0.009 | Age | 0.003 | 0.002 |
| | (0.010) | (0.027) | | (0.003) | (0.005) |
| High SES | -0.233 | -0.674 | Scholarship | -0.003 | -0.060 |
| | (0.174) | (0.522) | | (0.010) | (0.050) |
| GPA | -0.043 | -0.153 | Predicted callback rate | -0.644 | -0.136 |
| | (0.066) | (0.198) | | (0.504) | (0.888) |
| Business major | 0.008 | 0.010 | | | |
| | (0.008) | (0.021) | | | |
| Employment gap | 0.011 | 0.034 | | | |
| | (0.009) | (0.023) | | | |
| Current unemp.: 3+ | 0.013 | 0.005 | | | |
| | (0.012) | (0.032) | | | |
| 6+ | -0.008 | -0.038 | | | |
| | (0.012) | (0.029) | | | |
| 12+ | 0.001 | 0.021 | | | |
| | (0.012) | (0.032) | | | |
| Past unemp.: 3+ | 0.029 | 0.065 | | | |
| | (0.012) | (0.031) | | | |
| 6+ | -0.011 | -0.016 | | | |
| | (0.012) | (0.033) | | | |
| 12+ | -0.004 | 0.019 | | | |
| | (0.012) | (0.031) | | | |
| Predicted callback rate | 0.476 | -0.041 | | | |
| | (0.248) | (0.626) | | | |
| Joint $p$-value | | 0.452 | Joint $p$-value | | 0.589 |
| Sample size | | 9,220 | Sample size | | 6,416 |

Notes: This table reports results from tests of the assumption that applications at each job are independent trials. Estimates come from regressions of a callback indicator on a resume characteristic and the mean of this characteristic across other resumes at the same job. Columns (1)-(2) use data from Nunley et al. (2015), and columns (3)-(4) use data from Arceo-Gomez and Campos-Vasquez (2014). The predicted callback rate is the fitted value from a regression of a callback indicator on all resume characteristics, jackknifed at the job level. $P$-value comes from a test of the hypothesis that coefficients on the leave-out mean are zero for all individual characteristics. Standard errors, clustered at the job level, appear in parentheses.

Table III: Moments of callback rate distribution, BM data

| Moment | No constraints (1) | Shape constraints (2) |
|---|---|---|
| $E[p_w]$ | 0.094 (0.006) | 0.094 (0.007) |
| $E[p_b]$ | 0.063 (0.006) | 0.063 (0.006) |
| $E[(p_w - E[p_w])^2]$ | 0.040 (0.005) | 0.040 (0.004) |
| $E[(p_b - E[p_b])^2]$ | 0.023 (0.004) | 0.023 (0.003) |
| $E[(p_w - E[p_w])(p_b - E[p_b])]$ | 0.028 (0.004) | 0.028 (0.003) |
| $E[(p_w - E[p_w])^2(p_b - E[p_b])]$ | 0.015 (0.003) | 0.014 (0.002) |
| $E[(p_w - E[p_w])(p_b - E[p_b])^2]$ | 0.012 (0.003) | 0.012 (0.002) |
| $E[(p_w - E[p_w])^2(p_b - E[p_b])^2]$ | 0.010 (0.003) | 0.010 (0.002) |
| $J$-statistic: |  | 0.00 |
| $P$-value: |  | 1.000 |
| Sample size | 1,112 | |

Notes: This table reports generalized method of moments (GMM) estimates of moments of the joint distribution of job-specific white and black callback rates in the Bertrand and Mullainathan (2004) data. Estimates in column (2) come from a shape-constrained GMM procedure imposing that the moments are consistent with a well-defined probability distribution. The $J$-statistic is the minimized shape-constrained GMM criterion function. The $p$-value come from a bootstrap test of the hypothesis that the model restrictions are satisfied.

Table IV: Moments of callback rate distribution, Nunley et al. data

| Moment | Design-specific estimates | | | | Combined estimates |
| --- | --- | --- | --- | --- | --- |
| | (2,2) design (1) | (3,1) design (2) | (1,3) design (3) | P-value (4) | (5) |
| $E[p_w]$ | 0.174 (0.010) | 0.199 (0.025) | 0.142 (0.015) | 0.027 | 0.177 (0.007) |
| $E[p_b]$ | 0.148 (0.010) | 0.149 (0.015) | 0.157 (0.013) | 0.854 | 0.153 (0.007) |
| $E[(p_w - E[p_w])^2]$ | 0.089 (0.007) | 0.108 (0.009) | - | 0.097 | 0.095 (0.004) |
| $E[(p_b - E[p_b])^2]$ | 0.085 (0.007) | - | 0.083 (0.008) | 0.857 | 0.084 (0.004) |
| $E[(p_w - E[p_w])(p_b - E[p_b])]$ | 0.083 (0.006) | 0.084 (0.009) | 0.080 (0.009) | 0.926 | 0.084 (0.004) |
| $E[(p_w - E[p_w])^3]$ | - | 0.051 (0.008) | - | | 0.106 (0.006) |
| $E[(p_b - E[p_b])^3]$ | - | - | 0.044 (0.007) | | 0.092 (0.006) |
| $E[(p_w - E[p_w])^2(p_b - E[p_b])]$ | 0.044 (0.004) | 0.043 (0.007) | - | 0.875 | 0.040 (0.002) |
| $E[(p_w - E[p_w])(p_b - E[p_b])^2]$ | 0.047 (0.005) | - | 0.045 (0.007) | 0.819 | 0.042 (0.002) |
| $E[(p_w - E[p_w])^3(p_b - E[p_b])]$ | - | 0.034 (0.005) | - | - | 0.035 (0.002) |
| $E[(p_w - E[p_w])(p_b - E[p_b])^3]$ | - | - | 0.037 (0.006) | - | 0.037 (0.002) |
| $E[(p_w - E[p_w])^2(p_b - E[p_b])^2]$ | 0.036 (0.004) | - | - | - | 0.038 (0.002) |
| J-statistic: | | | | | 23.09 |
| P-value: | | | | | 0.190 |
| Sample size | 1,146 | 544 | 550 | | 2,240 |

Notes: This table reports generalized method of moments (GMM) estimates of moments of the joint distribution of job-specific white and black callback rates in the Nunley et al. (2015) data. Columns (1), (2), and (3) show estimates based on jobs that received 2 white and 2 black, 3 white and 1 black, and 1 white and 3 black applications, respectively. Estimates in column (4) come from a shape-constrained GMM procedure imposing that the moments are consistent with a well-defined probability distribution. The J-statistic is the minimized shape-constrained GMM criterion function. The p-value come from a bootstrap test of the hypothesis that the model restrictions are satisfied.

Table V: Moments of callback rate distribution, AGCV data

| Moment | No constraints (1) | Shape constraints (2) | Moment | No constraints (3) | Shape constraints (4) |
|---|---|---|---|---|---|
| $E[p_f]$ | 0.138 (0.010) | 0.140 (0.009) | $E\left[(p_f - E[p_f])^4\right]$ | 0.023 (0.004) | 0.025 (0.003) |
| $E[p_m]$ | 0.108 (0.009) | 0.114 (0.009) | $E[(p_m - E[p_m])^4]$ | 0.019 (0.004) | 0.024 (0.003) |
| $E\left[(p_f - E[p_f])^2\right]$ | 0.066 (0.009) | 0.066 (0.005) | $E\left[(p_f - E[p_f])^4 (p_m - E[p_m])\right]$ | 0.012 (0.003) | 0.011 (0.002) |
| $E[(p_m - E[p_m])^2]$ | 0.048 (0.005) | 0.054 (0.005) | $E[(p_f - E[p_f])(p_m - E[p_m])^4]$ | 0.013 (0.003) | 0.013 (0.002) |
| $E[(p_f - E[p_f])(p_m - E[p_m])]$ | 0.043 (0.005) | 0.044 (0.004) | $E\left[(p_f - E[p_f])^3 (p_m - E[p_m])^2\right]$ | 0.012 (0.003) | 0.011 (0.002) |
| $E\left[(p_f - E[p_f])^3\right]$ | 0.025 (0.005) | 0.063 (0.007) | $E\left[(p_f - \mu_f)^2 (p_m - E[p_m])^3\right]$ | 0.012 (0.003) | 0.012 (0.002) |
| $E[(p_m - E[p_m])^3]$ | 0.031 (0.005) | 0.050 (0.007) | $E\left[(p_f - E[p_f])^4 (p_m - E[p_m])^2\right]$ | 0.010 (0.002) | 0.009 (0.001) |
| $E\left[(p_f - E[p_f])^2 (p_m - E[p_m])\right]$ | 0.020 (0.004) | 0.017 (0.002) | $E\left[(p_f - E[p_f])^2 (p_m - E[p_m])^4\right]$ | 0.010 (0.002) | 0.009 (0.001) |
| $E[(p_f - E[p_f])(p_m - E[p_m])^2]$ | 0.022 (0.004) | 0.020 (0.002) | $E\left[(p_f - E[p_f])^3 (p_m - E[p_m])^3\right]$ | 0.009 (0.002) | 0.009 (0.001) |
| $E\left[(p_f - E[p_f])^3 (p_m - E[p_m])\right]$ | 0.015 (0.003) | 0.015 (0.002) | $E\left[(p_f - E[p_f])^4 (p_m - E[p_m])^3\right]$ | 0.008 (0.002) | 0.007 (0.001) |
| $E[(p_f - E[p_f])(p_m - E[p_m])^3]$ | 0.016 (0.003) | 0.017 (0.002) | $E\left[(p_f - E[p_f])^3 (p_m - E[p_m])^4\right]$ | 0.008 (0.002) | 0.008 (0.001) |
| $E\left[(p_f - E[p_f])^2 (p_m - E[p_m])^2\right]$ | 0.016 (0.003) | 0.016 (0.002) | $E\left[(p_f - E[p_f])^4 (p_m - E[p_m])^4\right]$ | 0.009 (0.002) | 0.006 (0.001) |

| | | |
|---|---|---|
| $J$-statistic: | | 3.33 |
| $P$-value: | | 0.790 |
| Sample size: | | 802 |

Notes: This table reports generalized method of moments (GMM) estimates of moments of the joint distribution of job-specific white and black callback rates in the Arceo-Gomez and Campos-Vasques (2014) data. Estimates in columns (2) and (4) come from a shape-constrained GMM procedure imposing that the moments are consistent with a well-defined probability distribution. The $J$-statistic is the minimized shape-constrained GMM criterion function. The $p$-value come from a bootstrap test of the hypothesis that the model restrictions are satisfied.

Table VI: Nonparametric estimates of treatment effect variation in resume correspondence studies

| | Bertrand & Mullainathan | | | Nunley et al. | | | Arceo-Gomez & Campos-Vasquez | | |
|---|---|---|---|---|---|---|---|---|---|
| | $p_b$ (1) | $p_w$ (2) | $p_b - p_w$ (3) | $p_b$ (4) | $p_w$ (5) | $p_b - p_w$ (6) | $p_m$ (7) | $p_f$ (8) | $p_m - p_f$ (9) |
| Mean | 0.063 (0.006) | 0.094 (0.007) | -0.031 (0.006) | 0.153 (0.007) | 0.177 (0.007) | -0.023 (0.005) | 0.114 (0.009) | 0.140 (0.009) | -0.025 (0.008) |
| Standard deviation | 0.152 (0.011) | 0.199 (0.011) | 0.082 (0.012) | 0.290 (0.008) | 0.308 (0.007) | 0.102 (0.009) | 0.231 (0.011) | 0.257 (0.010) | 0.179 (0.011) |
| Correlation with $p_w$ or $p_f$ | 0.927 (0.055) | 1.000 - | -0.717 (0.089) | 0.944 (0.018) | 1.000 - | -0.336 (0.048) | 0.735 (0.035) | 1.000 - | -0.483 (0.051) |
| Skewness | - | - | - | 3.757 (0.074) | 3.648 (0.087) | -4.450 (0.405) | 4.067 (0.140) | 3.748 (1.161) | -1.403 (0.385) |
| Excess kurtosis | - | - | - | - | - | - | 8.452 (1.458) | 5.756 (8.790) | 12.227 (2.291) |

Note: This table reports shape-constrained generalized method of moments (GMM) estimates of key features of the joint distribution of treatment and control callback rates in three resume correspondence studies. Columns (1)-(3) show estimates for black and white callback rates in Bertrand and Mullainathan (2004), columns (4)-(6) display estimates for black and white callback rates in Nunley et al. (2015), and columns (7)-(9) show estimates for male and female callback rates in Arceo-Gomez and Campos-Vasquez (2014). Standard error are computed using the numerical bootstrap procedure described by Hong and Li (2017).

## Table VII: Upper bounds on shares not discriminating, BM data

|  | $\Pr(p_w = p_b)$ | | | |
|---|---|---|---|---|
|  | Analytic | Lin. prog. | $\Pr(p_w \geq p_b)$ | $\Pr(p_w \leq p_b)$ |
| Callbacks | (1) | (2) | (3) | (4) |
| All | - | 0.870 | 1.000 | 0.870 |
| 0 | - | 0.962 | 1.000 | 0.962 |
| 1 | 0.637 | 0.576 | 1.000 | 0.576 |
| 2 | 0.621 | 0.558 | 1.000 | 0.558 |
| 3 | 0.560 | 0.492 | 1.000 | 0.492 |
| 4 | - | 0.788 | 1.000 | 0.788 |
| $J$-statistic: | 29.26 | | 0.00 | 29.26 |
| $P$-value: | 0.000 | | 1.000 | 0.000 |

Notes: This table reports upper bounds on the probability that jobs in the Bertrand and Mullainathan (2004) data do not discriminate based upon race. Column (1) shows upper bounds on the fraction of jobs that have the same callback rate for black and white applicants computed using the analytic formula in the text. Column (2) shows corresponding bounds computed by linear programming imposing constraints across callback strata. Column (3) and (4) show upper bounds on the fraction of jobs that are at least as likely to call white applicants as to call black applicants and vice versa, computed via linear programming. $J$-statistics and $p$-values come from bootstrap tests of the hypothesis that these probabilities equal one for all jobs.

Table VIII: Upper bounds on shares not discriminating, Nunley et al. data

| Design | Callbacks | $\Pr(p_w = p_b)$ Analytic (1) | $\Pr(p_w = p_b)$ Lin. prog. (2) | $\Pr(p_w \geq p_b)$ (3) | $\Pr(p_w \leq p_b)$ (4) |
|---|---|---|---|---|---|
| All | All | - | 0.642 | 0.846 | 0.827 |
| (2,2) | 0 | - | 0.848 | 0.907 | 0.952 |
|  | 1 | 0.824 | 0.328 | 0.815 | 0.567 |
|  | 2 | 0.718 | 0.309 | 0.984 | 0.325 |
|  | 3 | 0.874 | 0.179 | 0.933 | 0.264 |
|  | 4 | - | 0.579 | 0.743 | 0.872 |
| (3,1) | 0 | - | 0.853 | 0.898 | 0.964 |
|  | 1 | 0.861 | 0.337 | 0.894 | 0.549 |
|  | 2 | 0.738 | 0.332 | 0.998 | 0.336 |
|  | 3 | 0.816 | 0.151 | 0.922 | 0.251 |
|  | 4 | - | 0.566 | 0.767 | 0.837 |
| (1,3) | 0 | - | 0.839 | 0.916 | 0.936 |
|  | 1 | 0.810 | 0.323 | 0.754 | 0.594 |
|  | 2 | 0.855 | 0.326 | 0.958 | 0.369 |
|  | 3 | 0.913 | 0.204 | 0.955 | 0.262 |
|  | 4 | - | 0.581 | 0.723 | 0.893 |
|  | $J$-statistic: | | 62.64 | 23.46 | 62.64 |
|  | $P$-value: | | 0.000 | 0.120 | 0.000 |

Notes: This table reports upper bounds on the probability that jobs in the Nunley et al. (2015) data do not discriminate based upon race. Column (1) shows upper bounds on the fraction of jobs that have the same callback rate for black and white applicants computed using the analytic formula in the text. Column (2) shows corresponding bounds computed by linear programming imposing constraints across callback strata. Column (3) and (4) show upper bounds on the fraction of jobs that are at least as likely to call white applicants as to call black applicants and vice versa, computed via linear programming. $J$-statistics and $p$- values come from bootstrap tests of the hypothesis that these probabilities equal one for all jobs.

| | Table IX: Upper bounds on shares not discriminating, AGCV data | | | |
|---|---|---|---|---|
| | $\Pr(p_f = p_m)$ | | | |
| | Analytic | Lin. prog. | $\Pr(p_f \geq p_m)$ | $\Pr(p_f \leq p_m)$ |
| Callbacks | (1) | (2) | (3) | (4) |
| All | - | 0.723 | 0.911 | 0.812 |
| 0 | - | 0.864 | 0.960 | 0.905 |
| 1 | 0.968 | 0.105 | 0.586 | 0.520 |
| 2 | 0.460 | 0.284 | 0.740 | 0.544 |
| 3 | 0.495 | 0.424 | 0.953 | 0.472 |
| 4 | 0.522 | 0.497 | 0.945 | 0.553 |
| 5 | 0.687 | 0.654 | 0.829 | 0.825 |
| 6 | 0.662 | 0.591 | 0.788 | 0.803 |
| 7 | 0.875 | 0.514 | 0.843 | 0.671 |
| 8 | - | 0.924 | 0.989 | 0.935 |
| $J$-statistic: | | 369.66 | 33.88 | 359.95 |
| $P$-value: | | 0.000 | 0.005 | 0.000 |

Notes: This table reports upper bounds on the probability that jobs in the Arceo-Gomez and Campos-Vasquez (2014) data do not discriminate based upon sex. Column (1) shows upper bounds on the fraction of jobs that have the same callback rate for male and female applicants computed using the analytic formula in the text. Column (2) shows corresponding bounds computed by linear programming imposing constraints across callback strata. Column (3) and (4) show upper bounds on the fraction of jobs that are at least as likely to call female applicants as to call male applicants and vice versa, computed via linear programming. $J$-statistics and $p$-values come from bootstrap tests of the hypothesis that these probabilities equal one for all jobs.

Table X: Mixed logit estimates, Nunley et al. data

|  |  | Constant (1) | Types | |
|---|---|---|---|---|
|  |  |  | No selection (2) | Selection (3) |
| Distribution of logit($p_w$): | $\alpha_0$ | -4.708 (0.223) | -4.931 (0.242) | -4.927 (0.280) |
|  | $\sigma_\alpha$ | 4.745 (0.223) | 4.988 (0.249) | 4.983 (0.294) |
| Discrimination intensity: | $\beta_0$ | 0.456 (0.108) | 4.046 (1.563) | 4.053 (1.576) |
| Discrimination logit: | $\tau_0$ | - | -1.586 (0.416) | -1.556 (1.098) |
|  | $\tau_\alpha$ | - | - | -0.005 (0.180) |
| Fraction with $p_w \neq p_b$: |  | 1.000 | 0.168 | 0.170 |
| Log-likelihood |  | -2,792.1 | -2,788.2 | -2,788.2 |
| Parameters |  | 15 | 16 | 17 |
| Sample size |  | 2,305 | 2,305 | 2,305 |

Notes: This table reports estimates of logit models for callback probabilities in the Nunley et al. (2015) data. All models include a normally distributed job-specific intercept. Column (2) allows for two discrete types of jobs, one of which does not discriminate based upon race. Column (3) allows the probability of discrimination to depend on the job-specific intercept. Models are estimated by simulated maximum likelihood using 1,000 Halton draws of the intercept for each job. All models include resume covariates which are de-meaned in the estimation sample. Robust standard errors in parentheses.


\end{document}